\documentclass[11pt]{article}\usepackage{amsfonts}
\usepackage{amsmath}
\usepackage{fullpage}
\usepackage{amssymb}
\usepackage{physics}
\usepackage[ colorlinks = true,
             linkcolor = blue,
             urlcolor  = blue,
             citecolor = red,
             anchorcolor = green,
]{hyperref}\usepackage{graphicx}

\usepackage{pgfplots}
\usepackage{tikz}
\usepackage{xcolor}
\usepackage{caption}
\definecolor{dgreen}{HTML}{006600}
\usepackage[caption=false]{subfig}
\providecommand{\U}[1]{\protect\rule{.1in}{.1in}}
\providecommand{\U}[1]{\protect\rule{.1in}{.1in}}
\newcommand{\Abar}{\bar{A}}
\newcommand{\Bbar}{\bar{B}}

\newcommand{\rhoabxy}{\rho_{\Abar\Bbar XY}}
\newcommand{\rhoabxye}{\rho_{\Abar\Bbar XYE}}
\newcommand{\xx}{\tilde{x}}
\newcommand{\at}{\tilde{a}}

\def\>{\rangle}
\def\<{\langle}
\def\({\left(}
\def\){\right)}
\def\[{\left[}
\def\]{\right]}

\newtheorem{theorem}{Theorem}

\newtheorem{corollary}[theorem]{Corollary}

\newtheorem{definition}[theorem]{Definition}

\newtheorem{notation}[theorem]{Notation}

\newtheorem{proposition}[theorem]{Proposition}

\newtheorem{Property}[theorem]{Property}
\newenvironment{proof}[1][Proof]{\noindent\textbf{#1.} }{\ \rule{0.5em}{0.5em}}
\begin{document}
\title{Fundamental limits on key rates in device-independent\\quantum key distribution}
\author{ Eneet Kaur\thanks{Hearne Institute for Theoretical Physics, Department of Physics and Astronomy, and Center for Computation and Technology, Louisiana State University, Baton Rouge, Louisiana 70803, USA}  \and Mark M. Wilde\footnotemark[1] \and Andreas Winter\thanks{Departament de F\'{i}sica: Grup d'Informaci\'{o} Qu\'{a}ntica,
Universitat Aut\'{o}noma de Barcelona, ES-08193 Bellaterra (Barcelona), Spain}}

\date{\today}
\maketitle

\begin{abstract}
   In this paper, we introduce intrinsic non-locality and quantum intrinsic non-locality as quantifiers for Bell non-locality, and we prove that they satisfy certain desirable properties such as faithfulness, convexity, and monotonicity under local operations and shared randomness. We then prove that intrinsic non-locality is an upper bound on the secret-key-agreement capacity of any device-independent protocol conducted using a device characterized by a correlation $p$, while quantum intrinsic non-locality is an upper bound on the same capacity for a correlation arising from an underlying quantum model. We also prove that intrinsic steerability is faithful, and it is an upper bound on the secret-key-agreement capacity of any one-sided device-independent protocol conducted using a device characterized by an assemblage $\hat{\rho}$. 
   Finally, we  prove that quantum intrinsic non-locality is bounded from above by intrinsic steerability. 
\end{abstract}

\section{Introduction}
In principle, quantum key distribution (QKD) \cite{BB84,E91,SPCDLP2009} provides unconditional security \cite{Mayers2001,SP2000,LCT2015} for establishing secret key at a distance. In the standard QKD setting, Alice and Bob (two spatially separated parties) trust the functioning of their devices. That is, it is assumed that they know the 
ensemble of states that their sources are preparing and the measurement that their devices are performing. However, this is a very strong set of assumptions.  

It is possible to consider other scenarios in which the trust assumptions are relaxed while still obtaining unconditional security. When one of the devices is untrusted, the protocol is referred to as one-sided device-independent (SDI) quantum key distribution \cite{Tomamichel2011,Branciard2012}. If both the devices are untrusted, then we are dealing with the scenario of device-independent (DI) quantum key distribution \cite{Mayers98, Acin2007,Vazirani2014,Rotem2018}. 

It is interesting to note that the three scenarios of QKD mentioned above  are in correspondence with a hierarchy of quantum correlations \cite{Wiseman}. The standard QKD approach requires that Alice and Bob share entanglement \cite{HHHH2009} or that they are connected by a channel that can preserve entanglement. In SDI-QKD, a requirement for Alice and Bob to generate secret key is that their systems violate a steering inequality \cite{Branciard2012,CS2017}. For DI-QKD, Alice and Bob's systems should violate a Bell inequality \cite{CSHS69,Acin2007, BCPSW2014}. 

In this paper, we establish upper bounds on secret-key rates that are achievable with DI-QKD and SDI-QKD.  To this end, we first introduce intrinsic non-locality and quantum intrinsic non-locality as quantifiers of non-local correlations. We prove that they fulfill several desirable properties, such as monotonicity under local operations and shared randomness, convexity, faithfulness, superadditivity, and additivity with respect to tensor products. We also provide a proof for faithfulness of restricted intrinsic steerability, a quantifier of quantum steering introduced in \cite{Kaur2016}, thus solving an open question from \cite{Kaur2016}. 

Next, we consider a device that is characterized by a correlation $p$, and we allow Alice and Bob to perform local operations and public communication on its inputs and outputs (this contains the parameter estimation, error correction, and privacy amplification) to extract a secret key from this device. Then, we prove that intrinsic non-locality is an upper bound on the rate at which secret key can be extracted from this device, such that the secret key is protected from a third party possessing an arbitrary no-signaling extension of the correlation, as well as copies of all of the classical data publicly exchanged in the protocol. We do the same for quantum intrinsic non-locality and a third party possessing an arbitrary quantum extension of the correlation.

We then consider a device that is characterized by an assemblage $\hat{\rho}$ and prove that restricted intrinsic steerability is an upper bound on the rate at which secret key can be extracted from this device, such that the secret key is protected from a third party possessing an arbitrary no-signaling extension of the assemblage (as considered in \cite{Kaur2016}), as well as copies of all of the classical data publicly exchanged during the protocol. 

The present work is inspired by  \cite{Maurer1999}, which introduced intrinsic information and proved that it is an upper bound on the distillable secret key  for Alice and Bob protected from an adversary Eve, such that Alice has access to a random variable $X$, Bob to a random variable $Y$, and Eve to $Z$, such that the joint distribution is $P_{XYZ}$. Later, \cite{Christandl2004},  taking inspiration from the underlying idea of intrinsic
 information, defined squashed entanglement as a quantum version of
 the former, which turns out to be an entanglement measure with many 
 desirable properties. (See also \cite{Tucci2002} for discussions related to squashed entanglement.) The squashed entanglement was later established as an upper bound on the distillable secret key of a bipartite quantum state \cite{CEH07} (see also \cite{Wilde2016} in this context). The squashed entanglement of a channel was later defined and proved to be an upper bound on the secret-key-agreement capacity of a quantum channel \cite{Takeoka2014a,Wilde2016}. Both the
 intrinsic steerability from \cite{Kaur2016}, intrinsic non-locality, and the quantum intrinsic non-locality
 defined here are strongly related to these previous quantities. It
 is fair to say that intrinsic non-locality is closest
 in spirit to \cite{Maurer1999}, in that, it is defined entirely in terms
 of classical random variables accessible to Alice and Bob.

This paper is structured as follows: we recall the definition of restricted intrinsic steerability in Section~\ref{sec:restricted_intrinsic_steerability}.  We then introduce intrinsic non-locality and quantum intrinsic non-locality, and we analyze its mathematical properties in Section~\ref{section:intrinsic-non-locality}. In Section~\ref{sec:faithfulness_steerability}, we provide a proof for the faithfulness of intrinsic steerability. Section~\ref{sec:faithfulness_nonlocality} provides a proof for the faithfulness of restricted intrinsic non-locality. We prove upper bounds on secret-key-agreement capacities for device-independent and one-sided-device-independent protocols in Section~\ref{section:upper_bounds}. In Section~\ref{section:examples}, we showcase our bounds for some specific examples,
thus obtaining explicit bounds on the secret-key rate that can be obtained from  specific bipartite correlations studied  in the device-independent literature. We end with Section~\ref{section:conclusion}, where we conclude and discuss some open questions. 

Note: After posting the first version of this paper to the arXiv, we became aware of related results presented in \cite{WDH19}. In \cite{WDH19}, squashed non-locality was introduced as a measure of non-locality. This quantity was then proven to be an upper bound on device-independent secret key rates that are secure against a no-signaling adversary with classical inputs and outputs. This is the scenario in which an untrusted no-signaling device is 
shared by the honest parties (as in our model), while the inputs and outputs of an adversary Eve are assumed to satisfy only the no-signaling conditions. The latter is more `liberal' than the models considered in our work.

\section{Restricted intrinsic steerability} \label{sec:restricted_intrinsic_steerability}

In this section, we recall the definition of restricted intrinsic steerability, which was introduced in \cite{Kaur2016}. We begin by recalling the notion of an assemblage. Let $\rho_{AB}$ be a bipartite quantum state shared
by Alice and Bob. Suppose that Alice performs a measurement labeled by
$x\in\mathcal{X}$, with $\mathcal{X}$ denoting a finite set of quantum
measurement choices, and she gets a classical output $a \in\mathcal{A}$, with
$\mathcal{A}$ denoting a finite set of measurement outcomes. An
\textit{assemblage} \cite{Pusey2013} consists of the state of Bob's subsystem and the
conditional probability of Alice's outcome $a$ (correlated with Bob's state)
given the measurement choice $x$. This is specified as $\{ p_{\bar{A}%
|X}(a|x),\rho_{B}^{a,x}\} _{a\in\mathcal{A},x\in\mathcal{X}}$. The
sub-normalized state possessed by Bob is $\hat{\rho}_{B}^{a,x}:=p_{\bar{A}%
|X}(a|x)\rho_{B}^{a,x}$
. Taking
$p_{X}(x)$ as a probability distribution over measurement choices, we can then
embed the assemblage $\{ \hat{\rho}_{B}^{a,x}\}_{a,x}$ in a classical-quantum
state as follows:
\begin{equation}
\rho_{X\bar{A}B}:=\sum_{a,x}p_{X}(x) \[x\,a\]_{X\bar{A}} \otimes\hat{\rho}_{B}^{a,x}.%
\end{equation}
\begin{notation}
In the above and what follows, we employ the shorthand $[x\,a]_{X\Abar}$ to denote $\op{x}_X\otimes\op{a}_{\bar{A}}$. 
\end{notation}

Assemblages are restricted by the no-signaling principle. That is, the reduced
state of Bob's system should not depend on the input $x$ to Alice's black box
if the measurement output $a$ is not available to him:
\begin{equation}
\sum_{a}\hat{\rho}_{B}^{a,x}=\sum_{a}\hat{\rho}_{B}^{a,x^{\prime}}
\quad\forall x,x^{\prime}\in\mathcal{X}.
\end{equation}
This is equivalent to $I(X;B)_{\rho}=0$ for all input probability distributions $p_X(x)$, where $I(X;B)_{\rho}:=H(X)_{\rho
}+H(B)_{\rho}-H(XB)_{\rho}$ is the mutual information of the reduced state
$\rho_{XB}=\operatorname{Tr}_{\bar{A}}(\rho_{X\bar{A}B})$.

An assemblage is referred to as LHS (local-hidden-state) if it arises from a classical shared hidden variable $\Lambda$ in the following sense:
\begin{equation}
\hat{\rho}_B^{a,x}:= \sum_{\lambda}p_{\Lambda}(\lambda)p_{\bar{A}|X\Lambda}(a|x,\lambda)\rho^{\lambda}_B.
\end{equation}

We now recall a measure of steerability that was introduced in \cite{Kaur2016}:
\begin{definition}
[Restricted intrinsic steerability \cite{Kaur2016}]\label{def:reducedsteering-CMI} Let $\{\hat{\rho}_{B}^{a,x}\}_{a,x}$ denote an assemblage, and let $\rho_{X\bar
{A}B}$ denote a corresponding classical--quantum state. 
Consider a no-signaling extension
$\rho_{X\bar{A}BE}$ of $\rho_{X\bar{A}B}$ of the following form: 
\begin{equation}
\rho_{X\bar{A}B E}:=\sum_{a,x}p_{X}(x)\[x\,a\]_{X\bar{A}}\otimes\hat{\rho}_{BE}^{a,x}, \label{eq:rext-form}
\end{equation}
where $\hat{\rho}_{BE}^{a,x}$ satisfies $\operatorname{Tr}_{E}(\hat{\rho}%
_{BE}^{a,x})=\hat{\rho}_{B}^{a,x}$ and the following no-signaling constraints: 
\begin{equation}
\sum_{a}\hat{\rho}_{BE}^{a,x}=\sum_{a}\hat{\rho}_{BE}^{a,x^{\prime}%
}\quad \forall x,x^{\prime}\in\mathcal{X} . \label{eq:no-sig-extension-RIS}%
\end{equation}
The restricted intrinsic steerability  of $\{\hat{\rho}_{B}%
^{a,x}\}_{a,x}$ is defined as follows:%
\begin{equation}
S(\bar{A};B)_{\hat{\rho}}:=\sup_{p_{X}}\inf_{\rho_{X\bar{A}BE}}I(X\bar
{A};B|E)_{\rho},
\end{equation}
where the supremum is with respect to all probability distributions $p_{X}$
and the infimum is with respect to all non-signaling extensions of
$\rho_{X\bar{A}B}$ as specified above.
Furthermore, the conditional mutual information of a tripartite state $\sigma_{KLM}$ is defined as
\begin{equation}
I(K;L|M)_{\sigma}:=H(KM)_{\sigma} + H(LM)_{\sigma} - H(M)_{\sigma} - H(KLM)_{\sigma}.    
\end{equation}
Using the no-signaling constraints, which imply that $I(X;B|E)_{\rho}=0$, and the chain rule for conditional mutual information, it follows that 
\begin{equation}
S(\bar{A};B)_{\hat{\rho}}:=\sup_{p_{X}}\inf_{\rho_{X\bar{A}BE}}I(\bar
{A};B|EX)_{\rho}.\label{eq:alt-RIS}
\end{equation}
\end{definition}

\section{Quantum non-locality}\label{section:intrinsic-non-locality}

\subsection{Correlations}

Consider a two-component device that takes in two inputs and gives out two outputs. Let one component be with Alice and the other component be with Bob. Let us set some notation now. Alice's component takes in an input letter $x \in \mathcal{X}$ and outputs $a \in \mathcal{A}$. Similarly, Bob's component accepts an input letter $y \in \mathcal{Y}$ and outputs $b \in \mathcal{B}$. We consider $\mathcal{X}$ and $\mathcal{Y}$ to be finite sets of quantum measurement choices and $\mathcal{A}$ and $\mathcal{B}$ to be finite sets of measurement outcomes. For simplicity, we consider $\mathcal{X}=\mathcal{Y}=\[s\]$ and $\mathcal{A}=\mathcal{B}=\[r\]$. The conditional probability distribution
$\{p(a,b|x,y)\}_{a,b \in [r], x,y \in [s]}$
corresponding to  the device is traditionally called a ``correlation.''  
Then the correlations can be divided as follows according to the constraints that they fulfill.
\begin{itemize}
    \item \textbf{Local correlations}: A correlation is said to have a local-hidden variable (LHV) description or be a local correlation if it can be written as
\begin{equation}\label{eq:local_boxes}
p(a,b|x,y)=\sum_{\lambda}p_{\Lambda}(\lambda)p(a|x,\lambda) p(b|y,\lambda),
\end{equation}
where $\Lambda$ is a local hidden variable, $p_{\Lambda}(\lambda)$ is the probability that the realization $\lambda$ of the local hidden variable $\Lambda$ occurs, $p(a|x,\lambda)$ is the probability of obtaining the outcome $a$ given $x$ and $\lambda$, and $p(b|y,\lambda)$ is the probability of obtaining the outcome $b$ given $y$ and $\lambda$. Let $\textbf{L}$ denote the set of correlations that can be written as in \eqref{eq:local_boxes}. A device characterized by local correlations is known as a local box. 
\item \textbf{Quantum correlations}: The set $\textbf{Q}$ of quantum correlations corresponds to the set of correlations that can be written as
\begin{equation}\label{eqn:quantum-correlation}
    p(a,b|x,y)= \operatorname{Tr}([\Lambda_x^a\otimes \Lambda_y^b]\rho_{AB}),
\end{equation}
where $\rho_{AB}$ is a bipartite quantum state and $\{\Lambda^a_x\}_a$ and $\{\Lambda^b_y\}_b$ are POVMs characterizing Alice and Bob's respective measurements with $\Lambda^a_x,\Lambda^b_y \geq 0$ for all $a\in\mathcal{A}$ and $b \in \mathcal{B}$ and  $\sum_a \Lambda^a_x=I$ and $\sum_b \Lambda^b_y=I$.
\item \textbf{No-signaling correlations}: The set \textbf{NS} corresponds to the set of correlations that fulfill the following no-signaling principle:
\begin{align}
\sum_{a}p(a,b|x,y)=\sum_a p(a,b|x',y)=p(b|y), \quad \forall x,x'\in [s]  \text{ and }  b\in [r],\, y \in [s].\label{Co1}\\
\sum_{b}p(a,b|x,y)=\sum_b p(a,b|x,y')=p(a|x), \quad \forall y,y' \in [s]  \text{ and }  a \in [r],\,x \in [s].\label{Co2}
\end{align}
 The no-signaling constraints \eqref{Co1} and \eqref{Co2} can be expressed
equivalently in terms of conditional mutual informations, namely
\begin{equation}
    \forall p(x,y)\quad I(X;\bar{B}|Y)_p = 0 = I(Y;\bar{A}|X)_p,
\end{equation}
with respect to the joint distribution $p(a,b,x,y)=p(x,y)p(a,b|x,y)$,
and where $p(x,y)$ ranges over probability distributions on $X$ and $Y$.
\end{itemize}
It is well known that local correlations are contained in the set of quantum correlations, that is, $\textbf{L} \subset \textbf{Q}$. Since the correlations in $\textbf{Q}$ fulfill the constraints in \eqref{Co1} and \eqref{Co2}, we have that $\textbf{Q}\subset \textbf{NS}$.  For more details on correlations, please refer to \cite{BCPSW2014}. 

An example of a correlation that belongs to the no-signaling correlations, but not the quantum correlations, is a Popescu-Rohrlich (PR) box \cite{popescu95} box, which is defined as follows: 
\begin{definition}[PR box]
A PR box is a device corresponding to the following correlation $p(a,b|x,y)$:
\begin{align}
p(0,0|x,y)=p(1,1|x,y)&=\frac{1}{2} \quad \textnormal{for} \quad(x,y) \neq (1,1),\nonumber\\
p(0,1|x,y)=p(1,0|x,y)&=\frac{1}{2} \quad \textnormal{for}\quad(x,y) = (1,1), \label{E4}
\end{align}
while $p(a,b|x,y)=0$ for all other quadruples. This correlation is no-signaling between Alice and Bob, as defined in \eqref{Co1} and \eqref{Co2}. 
\end{definition}

\subsection{Local operations and shared randomness}

Physically, local operations and shared randomness \cite{MWW2009,MW2011} refers to an operation in which Alice and Bob share unlimited free randomness between their two components and can perform local operations on
\begin{itemize}
\item the inputs given by Alice and Bob to their respective components,
\item the outputs of the two components to give the final outputs to Alice and Bob.
\end{itemize} 
The local operations and shared randomness act on the initial correlation
$p_i(a,b|x,y)$ corresponding to the device, in order  to yield a final, modified correlation $p_f(a,b|x,y)$. 
These operations can be parametrized as follows \cite{Gallego}:
\begin{equation} \label{eqn:LOSR}
p_f(a_f,b_f|x_f,y_f):=\sum_{a,b,x,y} O^{(L)}(a_f,b_f|a,b,x,y,x_f,y_f) p_i(a,b|x,y) I^{(L)}(x,y|x_f,y_f).
\end{equation}
Here, $I^{(L)}$ corresponds to a local correlation for a local device that takes in the inputs $x_f$ and $y_f$ from Alice and Bob, uses shared randomness, and performs local operations to yield  new inputs $x$ and $y$ for the main device characterized by $p_i$. This can be written as
\begin{align}
I^{(L)}(x,y|x_f,y_f)&=\sum_{\lambda_2}p_{\Lambda_2}(\lambda_2)I_A(x|x_f,\lambda_2)I_B(y|y_f,\lambda_2),\label{eqn:local_boxes2}
\end{align}
where $p_{\Lambda_2} (\lambda_2)$ corresponds to the probability distribution of the shared classical variable $\Lambda_2$, $I_A(x|x_f, \lambda_2)$ corresponds to the probability of obtaining $x$ given $x_f$ and $\lambda_2$, and $I_B(y|y_f, \lambda_2)$ corresponds to the probability of obtaining $y$ given $y_f$ and $\lambda_2$.

Once the initial device $p_i$ generates the outputs $a$ and $b$, it can be post-processed by a local device that is characterized by the local correlation $O^{(L)}$. This can be written as 
\begin{equation}
O^{(L)}(a_f,b_f|a,b,x,y,x_f,y_f)=\sum_{\lambda_1}p_{\Lambda_1}(\lambda_1)O_A(a_f|a,x,x_f,\lambda_1)O_B(b_f|b,y,y_f,\lambda_1). \label{eqn:local_boxes1}
\end{equation}
This device takes in $a,b,x,y,x_f,y_f$ and gives the final outputs $a_f,b_f$ by using shared randomness and performing local operations on the inputs. Here, $p_{\Lambda_1}(\lambda_1)$ is a probability distribution over the classical shared random variable $\lambda_1$, $O_A(a_f|a,x,x_f,\lambda_1)$ is a conditional probability distribution for obtaining $a_f$ given $x,x_f,\lambda_1,a$, and $O_B(b_f|b,y,y_f,\lambda_1)$ is a conditional probability distribution for obtaining $b_f$ given $y,y_f,\lambda_1,b$. See Figure~\ref{fig:LOSR} for a pictorial representation of the most general transformation of local operations and shared randomness on a correlation $p_i(a,b|x,y)$. 
\begin{figure}
    \centering
    \includegraphics[width=4in]{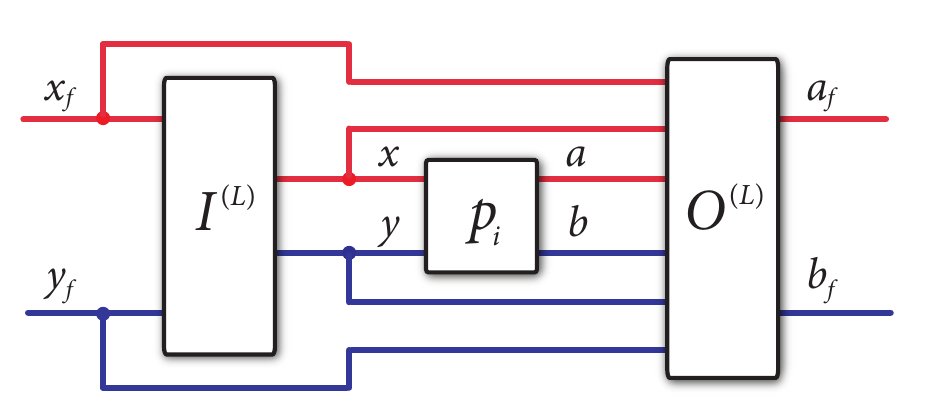}
    \caption{This figure depicts how local operations and shared randomness can act on an initial correlation $p_i(a,b|x,y)$ to produce a final correlation $p_f(a_f,b_f|x_f,y_f)$.}
    \label{fig:LOSR}
\end{figure}

In the resource theory of Bell non-locality \cite{JDV2014,Gallego}, the resources are  non-local correlations $p(a,b|x,y)$. Local operations and shared randomness are one possible set of free operations in this resource theory \cite{JDV2014}. It can be shown from the definition of a local correlation that the action of the local operations and share randomness transforms a local correlation to a correlation in~$\textbf{L}$. Furthermore, a quantum correlation remains in the set \textbf{Q} when acted upon by these free operations. To see this, replace the local boxes $O^{(L)}$ and $I^{(L)}$ in \eqref{eqn:LOSR} by separable states shared between Alice and Bob with the local states encoding the probability distributions required in \eqref{eqn:local_boxes2} and \eqref{eqn:local_boxes1} and the measurements as projective measurements. 

 In \cite{Gallego}, a larger set of free operations known as wirings and prior-to-input classical communication (WPICC) was considered. It was also shown in Lemma~6 of \cite{Gallego} that any quantifier that is monotone under local operations and shared randomness is also monotone under WPICCs.

\subsection{Intrinsic non-locality}

To calculate the amount of non-locality present in the correlation $p(a,b|x,y)$, we introduce a function $N : p(a,b|x,y)\rightarrow \mathbb{R}_{\geq 0}$, which we call \textit{intrinsic non-locality}. Consider a correlation $p(a,b|x,y) \in \textbf{NS}$. Now embed the correlation $p(a,b|x,y)$ into a classical-classical state as
\begin{equation}\label{eqn:classical_classical}
\rhoabxy:=\sum_{a,b,x,y}p(x,y)p(a,b|x,y)\left[a\,b\,x\,y\right]_{\Abar\Bbar XY},
\end{equation}
where $p(x,y)$ is a probability distribution for the measurement choices $x$ and $y$. 
Consider a no-signaling extension $\rhoabxye$ of $\rhoabxy$:
\begin{equation}
\rhoabxye:=\sum_{a,b,x,y}p(x,y)\left[a\,b\,x\,y\right]_{\Abar\Bbar XY}
\otimes p(a,b|x,y)\rho_E^{a,b,x,y},
\end{equation}
such that $\operatorname{Tr}_E(\rho_{\bar{A}\bar{B}XYE})= \rho_{\bar{A}\bar{B}XY}$, and the following no-signaling constraints hold:
\begin{align}
\sum_{a} p(a,b|x,y)\rho_E^{a,b,x,y}=\sum_{a} p(a,b|x',y)\rho_E^{a,b,x',y}\quad \forall x,x'\in \mathcal{X} , \, b \in \mathcal{B}, \, y \in \mathcal{Y}. \label{eq:constraint_1}
\end{align}
It is then easy to see that, given the value in system $Y$, the state of systems $X$ and systems $\bar{B}E$ is product. This is equivalent to the following constraint on conditional mutual information:
\begin{equation}\label{eqn:no_signaling_constraint}
I(\Bbar E;X|Y)_{\rho}=0 \quad \forall p(x,y).
\end{equation}
Similarly, the following no-signaling constraints hold
\begin{align}
\sum_{b} p(a,b|x,y)\rho_E^{a,b,x,y}=\sum_{b} p(a,b|x,y')\rho_E^{a,b,x,y'}\quad \forall y,y'\in \mathcal{Y}, \, a \in \mathcal{A}, \, x \in \mathcal{X}. \label{eq:constraint_2}
\end{align}
It is easy to see that, given the value in systems $X$, the state of systems $Y$ and $\bar{A}E$ is product. This is equivalent to the following constraint on conditional mutual information 
\begin{equation}
I(\Abar E;Y|X)_\rho=0  \quad \forall p(x,y). 
\end{equation}
Finally, we have that 
\begin{align}
\sum_{a,b} p(a,b|x,y)\rho_E^{a,b,x,y} &= \sum_{a,b} p(a,b|x',y)\rho_E^{a,b,x',y}\\
&= \sum_{a,b} p(a,b|x',y')\rho_E^{a,b,x',y'} \quad \forall x,x' \in \mathcal{X}, y,y'\in \mathcal{Y}. 
\end{align}
The first equality follows from \eqref{eq:constraint_1}, and the second equality follows from \eqref{eq:constraint_2}. This implies that the state of Eve's system is independent of the measurement choices, i.e., $I(XY;E)_\rho=0 $ for all $p(x,y)$. 
We can then quantify the amount of non-local correlations in the correlation $p(a,b|x,y)$ as $\inf_{\rho_{\bar{A}\bar{B}XYE}}I(\bar{A};\bar{B}|XYE)$, where the infimum is with respect to no-signaling extensions $\rho_{\Abar\Bbar XYE}$ of the above form. Since Alice and Bob want to maximize the non-local correlations of the two black boxes, we maximize over input probability distributions $p(x,y)$, leading us to the following definition:

\begin{definition}[Intrinsic non-locality]
The intrinsic non-locality of a correlation $p(a,b|x,y)\in \emph{\textbf{NS}}$ is defined as 
\begin{equation}
N(\Abar;\Bbar)_p = \sup_{p(x,y)}\inf_{\rhoabxye}I(\Abar;\Bbar|XYE)_\rho, 
\end{equation}
where $\rhoabxye$ is a no-signaling extension of the state $\rhoabxy$, i.e., subject to the constraints in \eqref{eq:constraint_1} and \eqref{eq:constraint_2}.
\end{definition}

\subsection{Quantum intrinsic non-locality}

 We now introduce a function $N^Q : p(a,b|x,y)\rightarrow \mathbb{R}_{\geq 0}$, which we call \textit{quantum intrinsic non-locality}, with $p(a,b|x,y)\in \textbf{Q}$. As stated above, the correlation in the set \textbf{Q} arises from some underlying state $\rho_{AB}$ and POVMs  of Alice and Bob characterized by $\left\{\Lambda^a_x\right\}_a$ and $\left\{\Lambda^b_y\right\}_b$, respectively.\footnote{For certain quantum correlations, it is possible to pinpoint the underlying quantum state and POVMs up to local isometries. See \cite{yang2013,Mayers2004} in this context.} Now, consider a quantum state $\rho_{ABE}$ such that $\operatorname{Tr}_E\left(\rho_{ABE}\right) = \rho_{AB}$. We call $\rho_{ABE}$ an extension of the state $\rho_{AB}$. Then, one possible extension of the classical-classical state $\rho_{\Abar\Bbar XY}$ as defined in \eqref{eqn:classical_classical}  is
\begin{align} \label{constraint_3}
    \rho_{\bar{A}\bar{B}XYE} &= \sum_{a,b,x,y} p(x,y) \operatorname{Tr}_{AB}[(\Lambda_x^a\otimes \Lambda_y^b\otimes I_E)\rho_{ABE}] \left[a\,b\,x\,y\right]_{\Abar\Bbar XY}, \\
&= \sum_{a,b,x,y} p(x,y) p(a,b|x,y) \[a\,b\,x\,y\]_{\Abar\Bbar XY} \otimes \rho_E^{a,b,x,y},
\end{align}
where $p(a,b|x,y)\rho_E^{a,b,x,y}:=\operatorname{Tr}_{AB}\!\left[\left(\Lambda_x^a\otimes \Lambda_y^b\otimes I_E\right)\rho_{ABE}\right]$. 
By definition, this extension is also a no-signaling extension and is subjected to the constraints in \eqref{eq:constraint_1} and \eqref{eq:constraint_2}. We call the extensions of the form in \eqref{constraint_3} \textbf{quantum extensions}.

For $p \in \textbf{Q}$, the set of no-signaling extensions of $p$ is strictly larger than the set of quantum extensions. For example, in the CHSH game, a correlation $p(a,b|x,y)$ reaching the Tsirelson bound only admits a trivial quantum extension, i.e., with 
constant $\rho^{a,b,x,y}_E$ independent of $a$, $b$, $x$, and $y$. Whereas, the no-signaling extensions of such a correlation are not extremal, as can be seen by
writing $p(a,b|x,y)$ as a convex combination of a PR box (with necessarily 
constant $\rho^{a,b,x,y}_E$ as an extension) and a local box (where $\rho_E^{a,b,x,y}$ contains the local hidden variable).

Therefore, to consider the regime in which there is an underlying quantum model, we define \textbf{quantum intrinsic non-locality} as follows:

\begin{definition}[Quantum intrinsic non-locality]
The quantum intrinsic non-locality of a correlation $p(a,b|x,y) \in \emph{\textbf{Q}}$ is defined as 
\begin{equation}
N^Q(\Abar;\Bbar)_p = \sup_{p(x,y)}\inf_{\rhoabxye}I(\Abar;\Bbar|XYE)_\rho, 
\end{equation}
where $\rhoabxye$ is a quantum extension of the state $\rhoabxy$ that is subject to the constraints in~\eqref{constraint_3}.
\end{definition}

\begin{proposition}\label{prop:relation}
If $p(a,b|x,y) \in \emph{\textbf{Q}}$, then \begin{equation}\label{eqn:inequality}
    N(\Abar;\Bbar)_p \leq N^Q(\Abar;\Bbar)_p.
\end{equation}
\end{proposition}
\begin{proof}
This follows from the observation that a quantum extension $\sigma_{\Abar\Bbar XYE}$ of $\rho_{\Abar\Bbar XY}$ is a particular kind of no-signaling extension.
\end{proof}

\subsection{Properties of intrinsic non-locality and quantum intrinsic non-locality}

In this section, we prove that intrinsic non-locality and quantum intrinsic non-locality are faithful, monotone with respect to local operations and shared randomness, superadditive, and additive with respect to tensor products of correlations. These are the properties that are desirable for a measure of Bell non-locality to possess. We also prove that the quantum intrinsic non-locality of a correlation is never larger than the intrinsic steerability of an associated assemblage. 

\begin{proposition} \label{prop:prof_faithful} Intrinsic non-locality and quantum intrinsic non-locality vanish for correlations
having a local hidden-variable model; i.e., if  $p(a,b|x,y)\in \emph{\textbf{L}}$, then $N(\Abar;\Bbar)_p=0$ and $N^Q(\Abar;\Bbar)_p=0$. 
\end{proposition}
\begin{proof} Given $p(a,b|x,y) \in \textbf{L}$, then we can write it as 
\begin{align}
p(a,b|x,y)&=\sum_{\lambda} p(\lambda) \,p(a|x,\lambda) \,p(b|y,\lambda). 
\end{align}
Embed this in a classical-classical state with $p(x,y)$ an arbitrary probability distribution over $x,y$: 
\begin{align}
\rhoabxy&=\sum_{a,b,x,y}p(x,y)\sum_{\lambda} p(\lambda) \,p(a|x,\lambda)\, p(b|y,\lambda)\left[a\,b\,x\,y\right]_{\Abar\Bbar XY}.
\end{align}
Then, consider the following quantum extension
\begin{align}
\rhoabxye&:=\sum_{a,b,x,y}p(x,y)\left[a\,b\,x\,y\right]_{\Abar\Bbar XY}\otimes \sum_{\lambda} p(\lambda) \,p(a|x,\lambda)\, p(b|y,\lambda) \,[\lambda]_{E}.
\end{align}
Then, by inspection, $\Abar$ and $\Bbar$ are independent given  $XYE$. This implies that $\inf_{\rho_{\Abar\Bbar XYE}}I(\Abar;\Bbar|XYE)_\rho=0$.  Since this equality holds for an arbitrary probability distribution $p(x,y)$, we can then conclude that $N^Q(\Abar;\Bbar)_p=0$. Then, by \eqref{eqn:inequality}, we conclude that $N(\Abar;\Bbar)_p=0$.
\end{proof}

\medskip
We later prove in Theorem \ref{theorem:faithfulness_nonlocality} that $N(\Abar;\Bbar)_{p}=0$ or $N^Q(\Abar;\Bbar)_p=0$ implies that $p \in \textbf{L}$. 

\medskip
We expect any quantifier of non-locality to be monotone under the free operations of local operations and shared randomness. That is, a free operation should not increase the amount of non-locality in the device.  We state this in the following proposition:
\begin{proposition}[Monotonicity of intrinsic non-locality] \label{prop:monotonicity}
Let $p_i(a,b|x,y)$ be a correlation, and let $p_f(a_f,b_f|x_f,y_f)$ be a correlation that results from the action of local operations and shared randomness on $p_i(a,b|x,y)$, so that we can write the final probability distribution as follows:
\begin{equation}
p_f(a_f,b_f|x_f,y_f):=\sum_{a,b,x,y} O^{(L)}(a_f,b_f|a,b,x,y,x_f,y_f) \,p_i(a,b|x,y)\, I^{(L)}(x,y|x_f,y_f),
\end{equation}
where $I^{(L)}(x,y|x_f,y_f)$ and $O^{(L)}(a_f,b_f|a,b,x,y,x_f,y_f)$ are local boxes as described in \eqref{eqn:local_boxes2} and \eqref{eqn:local_boxes1}. Then,
\begin{align}
N(\Abar;\Bbar)_{p_i}\geq N(\Abar_f;\Bbar_f)_{p_f}.
\end{align}
\end{proposition} 

\begin{proof}
 First, we embed $p_f(a_f,b_f|x_f,y_f)$ in a quantum state: 
\begin{equation}
\rho_{\Abar_f\Bbar_fX_fY_f}=\sum_{x_f,y_f,a_f,b_f}p(x_f,y_f)\,p_f(a_f,b_f|x_f,y_f)\,\[x_f\,y_f\,a_f\,b_f\]_{X_fY_f\bar{A}_f\Bbar_f}, \label{M1}
\end{equation}
where $p(x_f,y_f)$ is an arbitrary probability distribution for $x_f$ and $y_f$. 
Then invoking \eqref{eqn:LOSR}, \eqref{eqn:local_boxes2}, and \eqref{eqn:local_boxes1}, we obtain
\begin{multline}
\rho_{\Abar_f\Bbar_fX_fY_f}=\sum_{x_f,y_f,a_f,b_f}p(x_f,y_f)\sum_{a,b,x,y} \sum_{\lambda_2}p_{\Lambda_2}(\lambda_2)\,O_A(a_f|a,x_f,x,\lambda_2)\,O_B(b_f|b,y,y_f,\lambda_2)
\times \\p_i(a,b|x,y)
\sum_{\lambda_1}p_{\Lambda_1}(\lambda_1)\,I_A(x|x_f,\lambda_1)\,I_B(y|y_f,\lambda_1) \[x_f\,y_f\,a_f\,b_f\]_{X_fY_f\bar{A}_f\Bbar_f}. \label{M2}
\end{multline}
An arbitrary extension of the state in \eqref{M1} is given by
\begin{equation}
\rho_{\Abar_f\Bbar_fX_fY_fE}=\sum_{x_f,y_f,a_f,b_f}p(x_f,y_f)\,p_f(a_f,b_f|x_f,y_f)\[x_f\,y_f\,a_f\,b_f\]_{X_fY_f\bar{A}_f\Bbar_f}\otimes \rho_E^{a_f,b_f,x_f,y_f}. \label{eqn:ext1}
\end{equation}
A particular extension of the state in \eqref{M1} is given by
\begin{multline}
\zeta_{\Abar_f\Bbar_fX_fY_fE\Lambda_1\Lambda_2}=\sum_{x_f,y_f,a_f,b_f}p(x_f,y_f)\sum_{a,b,x,y} \sum_{\lambda_2}p_{\Lambda_2}(\lambda_2)\,O_A(a_f|a,x_f,x,\lambda_2) \times \,
\\O_B(b_f|b,y,y_f,\lambda_2)\, p_i(a,b|x,y) \times \\ \sum_{\lambda_1}p_{\Lambda_1}(\lambda_1)I_A(x|x_f,\lambda_1)\,I_B(y|y_f,\lambda_1)\[x_f\,y_f\,a_f\,b_f\]_{\bar{A}_f\Bbar_fX_fY_f}\otimes \tau_E^{a,b,x,y}\otimes \[\lambda_1\lambda_2\]_{\Lambda_1\Lambda_2}\label{M3}.
\end{multline}
This in turn is a marginal of the following state:
\begin{multline}
\zeta_{\Abar_f\Bbar_fX_fY_fE\Lambda_1\Lambda_2XY\Abar\Bbar}=\sum_{x_f,y_f,a_f,b_f}p(x_f,y_f)\sum_{a,b,x,y} \sum_{\lambda_2}p_{\Lambda_2}(\lambda_2)O_A(a_f|a,x_f,x,\lambda_2)\,O_B(b_f|b,y,y_f,\lambda_2)\times 
\\ p_i(a,b|x,y) \sum_{\lambda_1}p_{\Lambda_1}(\lambda_1)\,I_A(x|x_f,\lambda_1)\,I_B(y|y_f,\lambda_1)\[x_f\,y_f\,a_f\,b_f\]_{X_fY_f\bar{A}_f\Bbar_f}\otimes\\ \tau_E^{a,b,x,y} \otimes\[\lambda_1\lambda_2\]_{\Lambda_1\Lambda_2}\otimes [x\,y\,a\,b]_{XY\Abar\Bbar}.
\end{multline}
Consider that
\begin{align}
\inf_{\text{ext. in \eqref{eqn:ext1}}}I(\Abar_f;\Bbar_f|X_fY_fE)_\rho &\leq  I(\Abar_f;\Bbar_f|X_fY_fE\Lambda_1\Lambda_2)_\zeta\\&\leq I(\Abar X_fX\Lambda_2;\Bbar Y_fY\Lambda_2|X_fY_fE\Lambda_1\Lambda_2)_{\zeta}\\
&=I(\Abar X;\Bbar Y|X_fY_fE\Lambda_1\Lambda_2)_{\zeta}\\
&=I(\Abar X;\Bbar Y|X_fY_fE\Lambda_1)_{\zeta}\\
&=I(\Abar;\Bbar|XYX_fY_fE \Lambda_1)_{\zeta}+I(X;\Bbar|X_fY_fE\Lambda_1Y)_{\zeta}\nonumber\\&\qquad+I(Y;\Abar|X_fY_fE\Lambda_1X)_{\zeta}+I(X;Y|X_fY_f\Lambda_1E)_{\zeta}. \label{M4}
\end{align}
The first inequality follows from considering a particular extension in \eqref{M3}. The second inequality follows from data processing of conditional mutual information. The second equality follows because $\zeta_{\Abar\Bbar XYX_fY_fE\Lambda_1\Lambda_2}=\zeta_{\Abar\Bbar XYX_fY_fE\Lambda_1}\otimes \zeta_{\lambda_2}$. The last equality follows from the chain rule for conditional mutual information.
Now, let us consider each term in \eqref{M4}.
By inspection, 
\begin{multline}
\zeta_{\Abar\Bbar XYX_fY_fE\Lambda_1}= \sum_{x_f,y_f}p(x_f,y_f)\sum_{a,b,x,y,\lambda}p(\lambda_1) p_i(a,b|x,y)p(x,y|x_f,y_f,\lambda_1)\\ [x_f\,y_f\,\lambda_1 \,x \,y\,a\,b]_{X_fY_f\Lambda_1 XY\bar{A}\Bbar}\otimes \tau_E^{a,b,x,y}. 
\end{multline}
Upon re-arranging, we obtain
\begin{multline}
\zeta_{\Abar\Bbar XYX_fY_fE\Lambda_1}=\sum_{x,y}p(x,y) \sum_{x_f,y_f,\lambda_1}p(x_f,y_f,\lambda_1|x,y)\[x\,y\,x_f\,y_f\,\lambda_1\]_{XYX_fY_f\Lambda_1}\otimes
\\ \sum_{a,b}p_i(a,b|x,y)\tau_E^{a,b,x,y} \otimes [a\,b]_{\Abar\Bbar}.
\end{multline}
So, given $X,Y$, the states $\zeta_{\Abar\Bbar E}^{x,y}$ and $\zeta_{X_fY_f\Lambda_1}^{x,y}$ are in tensor product. Therefore $I(\Abar;\Bbar|XYX_fY_fE \Lambda_1)_{\zeta}=I(\Abar;\Bbar|XYE)_{\zeta}$, where $\zeta_{\Abar \Bbar XYE}$ is a no-signaling extension of $\rho_{\Abar\Bbar XY}$.
Now consider that
\begin{align}
&\zeta_{XX_fYY_f\Bbar E\Lambda_1}\notag \\
&=\sum_{x,y,x_f,y_f,\lambda_1}p(x,y,x_f,y_f,\lambda_1)\[x\,x_f\,y\,y_f\,\lambda_1\]_{XX_fYY_f\Lambda_1}\otimes\sum_b p(b|y)\,\tau_E^{b,y}\otimes\[b\]_{\bar{B}}.\\
&=\sum_y p(y)\[y\]_Y\otimes\!\sum_{x,x_f,y_f,\lambda_1} p(x_f,y_f,x,\lambda_1|y)\[x\,x_f\,y\,y_f\,\lambda_1\]_{XX_fYY_f\Lambda_1}\otimes \sum_b p(b|y)\, \tau_E^{b,y}\otimes \[b\]_{\Bbar}. 
\end{align}
Then, by inspection 
\begin{equation}
I(X;\Bbar|X_fY_fE\Lambda_1Y)_{\zeta}=0.
\end{equation}
Similarly, $I(Y;\Abar|Y_fX_fE\Lambda_1 X)_{\zeta}=0$.

Now, consider the term $I(X;Y|X_fY_fE\Lambda_1)_{\zeta}$, with
\begin{equation}
\zeta_{XYX_fY_fE\Lambda_1}:=\sum_{x_f,y_f}p(x_f,y_f)\sum_{x,y,\lambda_1} p(x|x_f,\lambda_1)\,p(y|y_f,\lambda_1)[x\,y\,x_f\,y_f\,\lambda_1]_{XYX_fY_f\Lambda_1} \otimes \rho_E. 
\end{equation}
Here, $X$ and $Y$ are independent given $X_f$, $Y_f$, and $\Lambda_1$. 
Therefore, $I(X;Y|X_fY_fE\Lambda_1)_{\zeta}=0$.
Combining the above equations, we obtain
\begin{equation}
\inf_{\text{ext. in (\ref{eqn:ext1})}}I(\Abar_f;\Bbar_f|X_fY_fE)_\rho \leq I(\Abar;\Bbar|XYE)_{\zeta}. \label{M5}
\end{equation}
Since \eqref{M5} is true for an arbitrary no-signaling extension of $\rho_{\Abar\Bbar XY}$, the above inequality holds after taking the infimum over all possible no-signaling extensions $\zeta_{\Abar\Bbar XYE}$.

Finally, we can take the supremum over all the measurement choices, and we find that
\begin{equation}
N(\Abar_f;\Bbar_f)_{p_f}\leq N(\Abar;\Bbar)_{p_i}.
\end{equation}
This concludes the proof.
\end{proof}

\begin{proposition}[Monotonicity of quantum intrinsic non-locality]
Let $p_i(a,b|x,y) \in \emph{\textbf{Q}}$, and let $p_f(a_f,b_f|x_f,y_f)$ result from the action of local operations and shared randomness on $p_i(a,b|x,y)$. We can write the final probability distribution as follows:
\begin{equation}
p_f(a_f,b_f|x_f,y_f):=\sum_{a,b,x,y} O^{(L)}(a_f,b_f|a,b,x,y,x_f,y_f) \,p_i(a,b|x,y)\, I^{(L)}(x,y|x_f,y_f),
\end{equation}
where $I^{(L)}(x,y|x_f,y_f)$ and $O^{(L)}(a_f,b_f|a,b,x,y,x_f,y_f)$ are local boxes as described in \eqref{eqn:local_boxes2} and \eqref{eqn:local_boxes1}. Then,
\begin{align}
N^Q(\Abar;\Bbar)_{p_i}\geq N^Q(\Abar_f;\Bbar_f)_{p_f}.
\end{align}
\end{proposition} 
\begin{proof}
First, we embed $p_f(a_f,b_f|x_f,y_f)$ in a quantum state: 
\begin{equation}
\rho_{\Abar_f\Bbar_fX_fY_f}=\sum_{x_f,y_f,a_f,b_f}p(x_f,y_f)\,p_f(a_f,b_f|x_f,y_f)\,\[x_f\,y_f\,a_f\,b_f\]_{X_fY_f\bar{A}_f\Bbar_f}, \label{MQ1}
\end{equation}
where $p(x_f,y_f)$ is an arbitrary probability distribution for $x_f$ and $y_f$. The set of quantum correlations $\textbf{Q}$ is closed under the action of local operations and shared randomness, implying that $p_f(a_f,b_f|x_f,y_f) \in \textbf{Q}$. Since $p_f(a_f,b_f|x_f,y_f)$ is also a quantum correlation, we know that there exists an underlying state $\sigma_{AB}$ and POVMs  $\left\{\Lambda_{x_f}^{a_f}\right\}_{a_f}$ and $\left\{\Lambda_{y_f}^{b_f}\right\}_{b_f}$, such that 
\begin{equation}
 p_f(a_f,b_f|x_f,y_f)= \operatorname{Tr}\left[\left(\Lambda_{x_f}^{a_f}\otimes \Lambda_{y_f}^{b_f}\right)\sigma_{AB}\right].
\end{equation}
An arbitrary quantum extension of the state in \eqref{MQ1} is given by
\begin{equation}
\sigma_{\Abar_f\Bbar_fX_fY_fE}=\sum_{x_f,y_f,a_f,b_f}p(x_f,y_f)\,p_f(a_f,b_f|x_f,y_f)\[x_f\,y_f\,a_f\,b_f\]_{X_fY_f\bar{A}_f\Bbar_f}\otimes \sigma_E^{a_f,b_f,x_f,y_f}, \label{eqn:ext}
\end{equation}
where
\begin{equation}
\sigma_E^{a_f,b_f,x_f,y_f}= \frac{1}{p_f(a_f,b_f|x_f,y_f)}\operatorname{Tr}_{AB}\left[\left(\Lambda_{x_f}^{a_f}\otimes \Lambda_{y_f}^{b_f}\otimes I_E\right)\sigma_{ABE}\right],
\end{equation}
 and $\sigma_{ABE}$ is an extension of $\sigma_{AB}$. 
Now, we know that
\begin{equation}
 p_f(a_f,b_f|x_f,y_f):=\sum_{a,b,x,y} O^{(L)}(a_f,b_f|a,b,x,y,x_f,y_f) \,p_i(a,b|x,y)\, I^{(L)}(x,y|x_f,y_f),
\end{equation}
and that the correlations $I^{(L)}(x,y|x_fy_f)$ and $O^{(L)}(a_f,b_f|a,b,x,y,x_f,y_f)$ are local correlations. Therefore, there exist separable states $\rho_{XY}$ and $\rho_{A_F B_F}$, along with POVMs that result in the correlations $I^{(L)}$ and $O^{(L)}$. That is,
\begin{align}
    I^{(L)}(x,y|x_f,y_f)&= \operatorname{Tr}\left[\left(\Lambda^x_{x_f}\otimes\Lambda^{y}_{y_f}\right)\rho_{XY}\right],\\
    O^{(L)}(a_f,b_f|a,b,x,y,x_f,y_f)&=\operatorname{Tr}\left[\left(\Lambda^{a_f}_{a,x_f,x}\otimes \Lambda^{b_f}_{b,b_f,y}\right)\rho_{A_FB_F}\right]
\end{align}
Furthermore, we know that the correlation $p_i(a,b|x,y)$ is a quantum correlation. Therefore, it has an underlying state $\rho_{AB}$ and POVMs characterized by $\left\{\Lambda_{x}^a\right\}_a$ and $\left\{\Lambda_{y}^b\right\}_b$.  Then
\begin{equation}
    p(a_f,b_f|x_f,y_f) = \sum_{a,b,x,y}\operatorname{Tr}\left[\left(\Lambda^{a_f}_{a,x_f,x}\otimes \Lambda^{b_f}_{b,b_f,y}\otimes \Lambda^a_x\otimes \Lambda^b_y\otimes \Lambda^x_{x_f}\otimes\Lambda^y_{y_f}\right)\left(\rho_{A_FB_F}\otimes \rho_{AB}\otimes \rho_{XY}\right)\right].
\end{equation}
  Since $\rho_{XY}$ is a separable state, we can write it as $\rho_{XY}=\sum_{\lambda_1}p(\lambda_1)\rho_X^{\lambda_1}\otimes \rho_Y^{\lambda_1}$. Let $\rho_{XY\Lambda_1}=\sum_{\lambda_1}p(\lambda_1)\rho_X^{\lambda_1}\otimes \rho_Y^{\lambda_1}\otimes \[\lambda_1\]_{\Lambda_1}$ be a particular extension of $\rho_{XY}$. Similarly, let $\rho_{A_FB_F\Lambda_2}$ be an extension of $\rho_{A_FB_F}$ and $\rho_{ABE}$ an extension of $\rho_{AB}$.

A particular quantum extension of the state in \eqref{MQ1} is given by
\begin{multline}
\rho_{\Abar_f\Bbar_fX_fY_fE\Lambda_1\Lambda_2}=\sum_{x_f,y_f,a_f,b_f}p(x_f,y_f)p_f(a_f,b_f|x_f,y_f)\[x_f,y_f,a_f,b_f\]_{X_fY_fA_fB_f}\rho_E^{a,b,x,y}\otimes \[\lambda_1 \lambda_2\]_{\Lambda_1\Lambda_2},
\end{multline}
where
\begin{equation}
\rho_E^{a,b,x,y} = \frac{1}{p(a,b|x,y)}\operatorname{Tr}_{AB}\[\left(\Lambda^a_x\otimes \Lambda^b_y\otimes I_E\right)\rho_{ABE}\].    
\end{equation}
Then it follows that 
\begin{multline}
\rho_{\Abar_f\Bbar_fX_fY_fE\Lambda_1\Lambda_2}= \sum_{x_f,y_f,a_f,b_f}p(x_f,y_f)\sum_{a,b,x,y} \sum_{\lambda_2}p_{\Lambda_2}(\lambda_2)\,O_A(a_f|a,x_f,x,\lambda_2) \times \,
\\O_B(b_f|b,y,y_f,\lambda_2)\, p_i(a,b|x,y) \times \\ \sum_{\lambda_1}p_{\Lambda_1}(\lambda_1)I_A(x|x_f,\lambda_1)\,I_B(y|y_f,\lambda_1)\[x_f\,y_f\,a_f\,b_f\]_{\bar{A}_f\Bbar_fX_fY_f}\otimes \rho_E^{a,b,x,y}\otimes \[\lambda_1\lambda_2\]_{\Lambda_1\Lambda_2}\label{MQ3}.
\end{multline}
This in turn is a marginal of the following state:
\begin{multline}
\rho_{\Abar_f\Bbar_fX_fY_fE\Lambda_1\Lambda_2XY\Abar\Bbar}=\sum_{x_f,y_f,a_f,b_f}p(x_f,y_f)\sum_{a,b,x,y} \sum_{\lambda_2}p_{\Lambda_2}(\lambda_2)O_A(a_f|a,x_f,x,\lambda_2)\,O_B(b_f|b,y,y_f,\lambda_2)\times 
\\ p_i(a,b|x,y) \sum_{\lambda_1}p_{\Lambda_1}(\lambda_1)\,I_A(x|x_f,\lambda_1)\,I_B(y|y_f,\lambda_1)\[x_f\,y_f\,a_f\,b_f\]_{X_fY_f\bar{A}_f\Bbar_f}\otimes\\ \rho_E^{a,b,x,y} \otimes\[\lambda_1\lambda_2\]_{\Lambda_1\Lambda_2}\otimes [x\,y\,a\,b]_{XY\Abar\Bbar}.
\end{multline}
Then, following arguments similar to that given in Proposition~\ref{prop:monotonicity}, we obtain
$N^Q(\Abar_f;\Bbar_f)_{p_f}\leq N^Q(\Abar;\Bbar)_{p_i}$.
\end{proof}

\begin{proposition}[Convexity of intrinsic non-locality] \label{prop:convexity_intrinsic}
Let $p(a,b|x,y)$ and $q(a,b|x,y)$ be two correlations, and let $\lambda \in [0,1]$. Let $t(a,b|x,y)$ be a mixture of the two correlations, defined as $t(a,b|x,y)=\lambda p(a,b|x,y)+\left(1-\lambda\right) q(a,b|x,y)$. 
Then
\begin{equation}
N(\Abar;\Bbar)_{t} \leq \lambda N(\Abar;\Bbar)_{p}+(1-\lambda) N(\Abar;\Bbar)_q.
\end{equation}
\end{proposition}
\begin{proof}
First, we embed the correlation $t(a,b|x,y)$ in the following classical-classical state~$\tau_{\bar{A}\bar{B}XY}$:
\begin{equation}
\tau_{\Abar \Bbar XY}:=\sum_{x,y,a,b} p(x,y)\, t(a,b|x,y)[x\,y\,a\,b]_{XY\Abar\Bbar}, \label{C1}
\end{equation}
where $p(x,y)$ is an arbitrary probability distribution. 
Similarly, embed $p(a,b|x,y)$ in $\rho_{\bar{A}\bar{B}XY}$ and $q(a,b|x,y)$ in $\gamma_{\Abar \Bbar XY}$:
\begin{align}
\rho_{\Abar \Bbar XY}:=\sum_{x,y,a,b} p(x,y)\,p(a,b|x,y)\[x\,y\,a\,b\]_{XY\Abar\Bbar}, \label{C3}\\
\gamma_{\Abar \Bbar XY}:=\sum_{x,y,a,b} p(x,y)\, q(a,b|x,y)\[x\,y\,a\,b\]_{XY\Abar\Bbar}. \label{C4}
\end{align}
Next, consider an arbitrary no-signaling extension of $\tau_{\bar{A}\bar{B}XY}$:
\begin{equation}
\tau_{\Abar \Bbar XYE}:=\sum_{x,y,a,b}p(x,y)\, t(a,b|x,y)\[x\,y\,a\,b\]_{XY\Abar\Bbar}\otimes \tau_E^{a,b,x,y}. \label{C2}
\end{equation}
Similarly, consider an arbitrary no-signaling extension of $\rho_{\bar{A}\bar{B}XY}$ and $\gamma_{\bar{A}\bar{B}XY}$:
\begin{align}
\rho_{\bar{A}\bar{B}XYE} = \sum_{x,y,a,b}p(x,y)\,p(a,b|x,y)\[x\,y\,a\,b\]_{XY\Abar\Bbar}\otimes \rho_E^{a,b,x,y}, \label{C5}\\
\gamma_{\bar{A}\bar{B}XYE} = \sum_{x,y,a,b}p(x,y)\,q(a,b|x,y)\[x\,y\,a\,b\]_{XY\Abar\Bbar}\otimes \gamma_E^{a,b,x,y}. \label{C6}
\end{align}
Now, consider the following particular no-signaling extension of $\tau_{\Abar\Bbar XY}$:
\begin{multline}
\zeta_{\Abar\Bbar XYEE'}:=\\
\sum_{x,y,a,b} p(x,y) \[x\,y\]_{XY} \otimes \left(\lambda \,p(a,b|x,y)\rho_E^{a,b,x,y}\otimes \[0\]_{E'}+(1-\lambda)\,q(a,b|x,y)\gamma_E^{a,b,x,y}\otimes \[1\]_{E'}\right).
\end{multline} 
Then, 
\begin{align}
\inf_{\text{ext. in \eqref{C2}}} I(\Abar;\Bbar|XYE)_{\tau}&\leq I(\Abar;\Bbar|XYEE')_{\zeta}\\
&= \lambda I(\Abar; \Bbar|XYE)_\rho+(1-\lambda) I(\Abar;\Bbar|XYE)_\gamma.
\end{align}
The first inequality follows from choosing a particular no-signaling extension. The equality follows from properties of conditional mutual information. 
Since this holds for all non-signaling extensions of the form in \eqref{C5} and \eqref{C6}, we conclude that
\begin{align}
\inf_{\text{ext. in } \eqref{C2}} I(\Abar;\Bbar|XYE)_{\zeta}&\leq \lambda \inf_{\text{ext. in } \eqref{C5}} I(\Abar;\Bbar|XYE)_{\rho} +(1-\lambda) \inf_{\text{ext. in } \eqref{C6} } I(\Abar;\Bbar|XYE)_{\gamma}.
\end{align}
Taking the supremum over all measurement choices, we find that 
\begin{multline}
\sup_{p(x,y)} \inf_{\text{ext. in } \eqref{C2}} I(\Abar;\Bbar|XYE)_{\zeta}\leq \lambda \sup_{p(x,y)} \inf_{\text{ext. in } (\ref{C5})} I(\Abar;\Bbar|XYE)_{\rho}+\\(1-\lambda) \sup_{p(x,y)}\inf_{\text{ext. in } \eqref{C6}} I(\Abar;\Bbar|XYE)_{\gamma}. 
\end{multline}
This completes the proof. 
\end{proof}

\begin{proposition}[Convexity of quantum intrinsic non-locality] \label{prop:quantum_convexity_intrinsic}
Let $p(a,b|x,y)$ and $q(a,b|x,y)$ be correlations in \emph{\textbf{Q}}, and let $\lambda \in [0,1]$. Let $t(a,b|x,y)$ be a mixture of the correlations defined as $t(a,b|x,y)=\lambda p(a,b|x,y)+\left(1-\lambda\right) q(a,b|x,y)$. 
Then
\begin{equation}
N^Q(\Abar;\Bbar)_{t} \leq \lambda N^Q(\Abar;\Bbar)_{p}+(1-\lambda) N^Q(\Abar;\Bbar)_q.
\end{equation}
\end{proposition}

\begin{proof}
Since \textbf{Q} is a convex set \cite{Pitowsky}, we know that $t(a,b|x,y) \in \textbf{Q}$. First, we embed the correlation $t(a,b|x,y)$ in the following quantum state~$\tau_{\bar{A}\bar{B}XY}$:
\begin{equation}
\tau_{\Abar \Bbar XY}:=\sum_{x,y,a,b} p(x,y)\, t(a,b|x,y)[x\,y\,a\,b]_{XY\Abar\Bbar}, \label{C1Q}
\end{equation}
where $p(x,y)$ is an arbitrary probability distribution. 
Similarly, embed $p(a,b|x,y)$ in $\rho_{\bar{A}\bar{B}XY}$ and $q(a,b|x,y)$ in $\gamma_{\Abar \Bbar XY}$:
\begin{align}
\rho_{\Abar \Bbar XY}:=\sum_{x,y,a,b} p(x,y)\,p(a,b|x,y)\[x\,y\,a\,b\]_{XY\Abar\Bbar}, \label{C3Q}\\
\gamma_{\Abar \Bbar XY}:=\sum_{x,y,a,b} p(x,y)\, q(a,b|x,y)\[x\,y\,a\,b\]_{XY\Abar\Bbar}. \label{C4Q}
\end{align}
Next, consider an arbitrary quantum extension of $\tau_{\bar{A}\bar{B}XY}$:
\begin{equation}
\tau_{\Abar \Bbar XYE}:=\sum_{x,y,a,b}p(x,y)\, t(a,b|x,y)\[x\,y\,a\,b\]_{XY\Abar\Bbar}\otimes \tau_E^{a,b,x,y}. \label{C2Q}
\end{equation}
Similarly, consider an arbitrary quantum extension of $\rho_{\bar{A}\bar{B}XY}$ and $\gamma_{\bar{A}\bar{B}XY}$:
\begin{align}
\rho_{\bar{A}\bar{B}XYE} = \sum_{x,y,a,b}p(x,y)\,p(a,b|x,y)\[x\,y\,a\,b\]_{XY\Abar\Bbar}\otimes \rho_E^{a,b,x,y}, \label{C5Q}\\
\gamma_{\bar{A}\bar{B}XYE} = \sum_{x,y,a,b}p(x,y)\,q(a,b|x,y)\[x\,y\,a\,b\]_{XY\Abar\Bbar}\otimes \gamma_E^{a,b,x,y}. \label{C6Q}
\end{align}
Let $\rho_{AB}$ be a quantum state that, along with the POVMs characterized by $\Lambda^a_{x}$ and $\Lambda^b_{y}$, yield the correlation $p(a,b|x,y)$. Let $\rho_{ABE}$ be an extension of $\rho_{AB}$. Similarly, let $\gamma_{AB}$ be a quantum state that, along with the POVMs characterized by $M^a_{x}$ and $M^b_{y}$, yield the correlation $q(a,b|x,y)$. Let $\gamma_{ABE}$ be an extension of $\gamma_{AB}$. Then, a particular quantum state that realizes the correlation $t(a,b|x,y)$ is the following:
\begin{align}
\tau_{ABA'B'}&=\lambda\rho_{AB} \otimes \op{00}_{A'B'}+(1-\lambda)\gamma_{AB}\otimes \op{11}_{A'B'}, \\
t(a,b|x,y)&=\operatorname{Tr}\left[\left(\Lambda^a_{x}\otimes\Lambda^b_{y}\otimes\left(\op{00}_{A'B'}\right)+M^a_{x}\otimes M^b_{y}\otimes\left(\op{11}_{A'B'}\right)\right)\left(\tau_{ABA'B'}\right)\right],
\end{align}
where it is understood that Alice is measuring $\sigma_Z$ on her system $A'$ and Bob is measuring $\sigma_Z$ on~$B'$, in addition to the other measurements on their systems $A$ and $B$.
Now, consider the following extension of $\tau_{ABA'B'}$:
\begin{equation}
\tau_{ABA'B'EE'}=\lambda\rho_{ABE} \otimes \op{000}_{A'B'E'}+(1-\lambda)\gamma_{ABE}\otimes \op{111}_{A'B'E'}.
\end{equation}
Furthermore, consider the following particular quantum extension of $\tau_{\Abar\Bbar XY}$:
\begin{multline}
\zeta_{\Abar\Bbar XYEE'}:=\\
\sum_{x,y,a,b} p(x,y) \[x\,y\]_{XY} \otimes \left(\lambda \,p(a,b|x,y)\rho_{E}^{a,b,x,y}\otimes \[0\]_{E'}+(1-\lambda)\,q(a,b|x,y)\gamma_{E}^{a,b,x,y}\otimes \[1\]_{E'}\right).
\end{multline} 
Then following similar arguments given in the proof of Proposition~\ref{prop:convexity_intrinsic}, we obtain
\begin{equation}
N^Q(\Abar;\Bbar)_{t} \leq \lambda N^Q(\Abar;\Bbar)_{p}+(1-\lambda) N^Q(\Abar;\Bbar)_q,
\end{equation}
concluding the proof.
\end{proof}
\begin{proposition}[Superadditivity and additivity of intrinsic non-locality]
\label{prop:additivity}
Let \\$p(a_1,a_2,b_1,b_2|x_1,x_2,y_1,y_2)$ be a correlation for which the following no-signaling constraints hold: 
\begin{align}
&\sum_{a_1}p(a_1,a_2,b_1,b_2|x_1,x_2,y_1,y_2)=\sum_{a_1}p(a_1,a_2,b_1,b_2|x_1',x_2,y_1,y_2) \quad \forall x_1',x_1 ,x_2,y_1,y_2\in [s],\,  a_2,b_1,b_2 \in [r],\nonumber\\
&\sum_{a_2}p(a_1,a_2,b_1,b_2|x_1,x_2,y_1,y_2)=\sum_{a_2}p(a_1,a_2,b_1,b_2|x_1,x_2',y_1,y_2) \quad \forall x_2',x_2 ,x_1,y_1,y_2\in [s],\,  a_1,b_1,b_2 \in [r],\nonumber \\
&\sum_{b_1}p(a_1,a_2,b_1,b_2|x_1,x_2,y_1,y_2)=\sum_{b_1}p(a_1,a_2,b_1,b_2|x_1,x_2,y_1',y_2) \quad \forall y_1',y_1 ,x_1,x_2,y_2\in [s],\,  a_1,a_2,b_2 \in [r], \nonumber\\
&\sum_{b_2}p(a_1,a_2,b_1,b_2|x_1,x_2,y_1,y_2)=\sum_{b_2}p(a_1,a_2,b_1,b_2|x_1,x_2,y_1,y_2') \quad \forall y_2',y_2 ,x_2,y_1,x_1\in [s],\,  a_1,a_2,b_1 \in [r]. \nonumber
\end{align}
Let $t(a_1,b_1|x_1,y_1)$ and $r(a_2,b_2|x_2,y_2)$ be correlations corresponding to the marginals of the probability distribution $p(a_1,a_2,b_1,b_2|x_1,x_2,y_1,y_2)$. Then the intrinsic non-locality is super-additive, in the sense that
\begin{equation}\label{eq:superadditivity}
N(\Abar_1\Abar_2;\Bbar_1\Bbar_2)_{p}\geq N(\Abar_1;\Bbar_1)_t+N(\Abar_2;\Bbar_2)_r.
\end{equation}
If $p(a_1,b_1,a_2,b_2|x_1,x_2,y_1,y_2)=t(a_1,b_1|x_1,y_1) r(a_2,b_2|x_2,y_2)$, then the intrinsic non-locality is additive in the following sense:
\begin{equation}
N(\Abar_1\Abar_2;\Bbar_1\Bbar_2)_p= N(\Abar_1;\Bbar_1)_t+N(\Abar_2;\Bbar_2)_r.
\end{equation}
\end{proposition}

\begin{proof} Consider the classical-classical state $\rho_{\Abar_1\Abar_2\Bbar_1\Bbar_2 X_1Y_1X_2Y_2}$ with the following arbitrary no-signaling extension:
\begin{multline} 
\rho_{\Abar_1\Abar_2\Bbar_1\Bbar_2X_1X_2Y_1Y_2E}=\sum_{x_1,x_2,y_1,y_2,a_1,a_2,b_1,b_2}p(x_1,y_1,x_2,y_2)\,p(a_1,b_1,a_2,b_2|x_1,x_2,y_1,y_2)\\ \[a_1\,b_1\,x_1\,y_1\,a_2\,b_2\,x_2\,y_2\]_{\Abar_1\Bbar_1X_1Y_1\Abar_2\Bbar_2X_2Y_2} \otimes \rho_E^{a_1,b_1,x_1,y_1,a_2,b_2,x_2,y_2}, \label{eq:no-sign-main}
\end{multline}
where $p(x_1,x_2,y_1,y_2)$ is an arbitrary probability distribution. From the chain rule of mutual information and non-negativity of conditional mutual information, we obtain
\begin{align}
&\!\!\!\!\!\!\!I(\Abar_1\Abar_2;\Bbar_1\Bbar_2|X_1X_2Y_1Y_2E)_{\rho}
\nonumber\\
&=I(\Abar_1\Abar_2;\Bbar_1|X_1Y_1X_2Y_2E)+I(\Abar_1\Abar_2;\Bbar_2|EX_1Y_1X_2Y_2\Bbar_1)\\ 
&= I(\Abar_1;\Bbar_1|X_1Y_1X_2Y_2E)_{\rho}+I(\Abar_2;\Bbar_1|EX_1Y_1X_2Y_2\Abar_1)_{\rho}\nonumber\\
&\qquad +I(\Abar_1;\Bbar_2|X_1Y_1X_2Y_2E\Bbar_1) + I(\Abar_2;\Bbar_2|X_1Y_1X_2Y_2E\Abar_1\Bbar_1)\\
&\geq I(\Abar_1;\Bbar_1|X_1Y_1X_2Y_2E)_{\rho}+I(\Abar_2;\Bbar_2|X_1Y_1X_2Y_2E\Abar_1\Bbar_1)_{\rho}.
\end{align}
From the no-signaling constraints in the statement of the proposition and \eqref{eq:no-sign-main}, we obtain
\begin{multline}\rho_{\Abar_1\Bbar_1X_1X_2Y_1Y_2E}=\sum_{a_1,b_1,x_1,x_2,y_1,y_2} p(x_1,x_2,y_1,y_2)\[a_1\,b_1\,x_1\,y_1\,x_2\,y_2\]_{\Abar_1\Bbar_2X_1Y_1X_2Y_2} \\ \otimes p(a_1,b_1|x_1,y_1)\,\rho_E^{x_1,y_1,a_1,b_1}.
\end{multline}
We first embed $t(a_1,b_1|x_1,y_1)$ in $\tau_{\Abar_1\Bbar_1X_1Y_1E}$, and $r(a_2,b_2|x_2,y_2)$ in $\gamma_{\Abar_2\Bbar_2X_2Y_2E}$ and consider the following arbitrary no-signaling extensions: 
\begin{align}
\tau_{\Abar_1\Bbar_1X_1Y_1E}&:=\sum_{x_1,y_1}p(x_1,y_1)\otimes\sum_{a_1,b_1}\[x_1\,y_1\,a_1\,b_1\]_{X_1Y_1\Abar_1\Bbar_1} \otimes t(a_1,b_1|x_1,y_1)\tau_E^{a_1,b_1,x_1,y_1}, \label{Su1}\\
\gamma_{\Abar_2\Bbar_2X_2Y_2E}&:=\sum_{x_2,y_2}p(x_2,y_2)\otimes\sum_{a_2,b_2}\[x_2\,y_2\,a_2\,b_2\]_{X_2Y_2\Abar_2\Bbar_2} \otimes r(a_2,b_2|x_2,y_2)\gamma_E^{a_2,b_2,x_2,y_2}.\label{S2}
\end{align}
Since $\rho_{\bar{A}_1\bar{B}_1X_1Y_1X_2Y_2E}$ is a particular no-signaling extension of $\tau_{\bar{A}_1\bar{B}_1X_1Y_1}$ and $\rho_{\bar{A}_1\bar{B}_1\bar{A}_2\bar{B}_2X_1Y_1X_2Y_2E}$ is a particular no-signaling extension of $\gamma_{\bar{A}_2\bar{B}_2X_2Y_2}$, we obtain the following inequality:
\begin{align}
&\!\!\!\!\!\!I(\Abar_1\Abar_2;\Bbar_1\Bbar_2|X_1X_2Y_1Y_2E)_{\rho}
\nonumber \\
&\geq I(\Abar_1;\Bbar_1|X_1Y_1X_2Y_2E)_{\rho}+I(\Abar_2;\Bbar_2|X_1Y_1X_2Y_2E\Abar_1\Bbar_1)_{\rho}\\&\geq \inf_{\textrm{ext. in } \eqref{Su1}}I(\Abar_1;\Bbar_1|X_1Y_1E)_{\tau}+\inf_{\textrm{ext. in } \eqref{S2}}I(\Abar_2;\Bbar_2|X_2Y_2E\Abar_1\Bbar_1)_{\gamma}.\label{eq:no_sig_addi}
\end{align}
Since \eqref{eq:no_sig_addi} holds for an arbitrary no-signaling extension of $\rho$, we obtain 
\begin{multline}
\inf_{\textrm {ext. in } \eqref{eq:no-sign-main}}I(\Abar_1\Abar_2;\Bbar_1\Bbar_2|X_1X_2Y_1Y_2E)_{\rho}\geq \\\inf_{\textrm{ext. in } \eqref{Su1}}I(\Abar_1;\Bbar_1|X_1Y_1E)_{\tau}+\inf_{\textrm{ext. in } \eqref{S2}}I(\Abar_2;\Bbar_2|X_2Y_2E\Abar_1\Bbar_1)_{\gamma}
\end{multline}
Since the above equation holds for arbitrary probability distributions, we can take a supremum over all probability distributions to obtain 
\begin{multline}
\sup_{p(x_1,y_1)p(x_2,y_2)} \inf_{\rho_{\Abar_1\Abar_2\Bbar_1\Bbar_2X_1X_2Y_1Y_2E}} I(\Abar_1\Abar_2;\Bbar_1\Bbar_2|X_1X_2Y_1Y_2E)_{\rho}\geq\\ \sup_{p(x_1,y_1)} \inf_{\tau_{\Abar_1\Bbar_1X_1Y_1E}}I(\Abar_1;\Bbar_1|X_1Y_1E)_{\tau} + \sup_{p(x_2,y_2)} \inf_{\gamma_{\Abar_2\Bbar_2X_2Y_2E}} I(\Abar_2;\Bbar_2|X_2Y_2E)_{\gamma}.
\end{multline}
Since we have considered a supremum over product probability distributions for the measurement choices on the LHS, we can relax this to consider the supremum over all probability distributions $p(x_1,y_1,x_2,y_2)$ of the measurement choices. 
This concludes the proof of (\ref{eq:superadditivity}). 

Now we give a proof for additivity of intrinsic non-locality with respect to product probability distributions. 
Since intrinsic non-locality is super-additive, it is sufficient to prove the  following sub-additivity property for product probability distributions:
\begin{equation}
N(\Abar_1\Abar_2;\Bbar_1\Bbar_2)_p\leq N(\Abar_1;\Bbar_1)_t+N(\Abar_2;\Bbar_2)_r.
\end{equation}
Consider the following states
\begin{multline}
\rho_{\Abar_1\Abar_2\Bbar_1\Bbar_2X_1X_2Y_1Y_2}=\sum_{a_1,b_1,x_1,y_1,a_2,b_2,x_2,y_2}p(x_1,x_2,y_1,y_2)\,t(a_1,b_1|x_1,y_1)\,r(a_2,b_2|x_2,y_2)\\ \[a_1\,b_1\,a_2\,b_2\,x_1\,x_2\,y_1\,y_2\]_{\Abar_1\Bbar_1\Abar_2\Bbar_2X_1Y_1X_2Y_2}.
\end{multline}
Consider an arbitrary extension of the state $\rho_{\Abar_1\Abar_2\Bbar_1\Bbar_2X_1X_2Y_1Y_2}$
\begin{multline}
\rho_{\Abar_1\Abar_2\Bbar_1\Bbar_2X_1X_2Y_1Y_2E}:=\sum_{a_1,b_1,x_1,y_1,a_2,b_2,x_2,y_2}p(x_1,x_2,y_1,y_2)\,t(a_1,b_1|x_1,y_1)\,r(a_2,b_2|x_2,y_2)\\\[a_1\,b_1\,x_1\,y_1\,a_2\,b_2\,x_2\,y_2\]\otimes \rho_E^{a_1,b_1,x_1,y_1,a_2,b_2,x_2,y_2}. \label{A1}
\end{multline}
Now, consider a particular extension of the state $\rho_{\Abar_1\Abar_2\Bbar_1\Bbar_2X_1X_2Y_1Y_2}$: 
\begin{multline}
\zeta_{\Abar_1\Abar_2\Bbar_1\Bbar_2X_1X_2Y_1Y_2E_1E_2}:=\sum_{a_1,b_1,x_1,y_1,a_2,b_2,x_2,y_2}p(x_1,x_2,y_1,y_2)\,t(a_1,b_1|x_1,y_1)\,r(a_2,b_2|x_2,y_2)\\\[a_1b_1a_2b_2x_1x_2y_1y_2\]_{\Abar_1\Bbar_1X_1Y_1\Abar_2\Bbar_2X_2Y_2}\otimes \rho_{E_1}^{a_1,b_1,x_1,y_1}\otimes\rho_{E_2}^{a_2,b_2,x_2,y_2}. \label{A2}
 \end{multline}
Then, we have the following set of inequalities:
\begin{align}
\inf_{\text{ext. in } \eqref{A1}}&I(\Abar_1\Abar_2;\Bbar_1\Bbar_2|X_1Y_1X_2Y_2E)_{\rho}\nonumber
\\&\leq I(\Abar_1\Abar_2;\Bbar_1\Bbar_2|X_1Y_1X_2Y_2E_1E_2)_{\zeta}\\&=
I(\Abar_1;\Bbar_1|X_1Y_1X_2Y_2E_1E_2)_{\zeta}+I(\Abar_2;\Bbar_1|E_1E_2X_1Y_1X_2Y_2\Abar_1)_{\zeta}\nonumber\\& \quad+I(\Abar_1;\Bbar_2|X_1Y_1X_2Y_2E_1E_2\Bbar_1)_{\zeta} + I(\Abar_2;\Bbar_2|X_1Y_1X_2Y_2E_1E_2\Abar_1\Bbar_1)_{\zeta}\\
&=I(\Abar_1;\Bbar_1|X_1Y_1X_2Y_2E_1E_2)_{\zeta}+I(\Abar_2;\Bbar_2|X_1Y_1X_2Y_2E_1E_2\Abar_1\Bbar_1)_{\zeta}.
 \end{align}
The first inequality follows from a particular choice of an extension. The first equality follows from the chain rule. For the second equality, observe the following:
\begin{align}
&I(\bar{A}_2;\Bbar_1|E_1E_2X_1Y_1X_2Y_2\Abar_1)_{\zeta}
\nonumber
\\&= H(\bar{A}_2|E_1E_2X_1Y_1X_2Y_2\Abar_1)_{\zeta}-H(\bar{A}_2|E_1E_2X_1Y_1X_2Y_2\Abar_1\Bbar_1)_{\zeta}\\
&= \sum_{x_1x_2y_1y_2}p(x_1,x_2,y_1,y_2)\left[H(\Abar_2|\Abar_1E_1E_2)_{\zeta^{x_1x_2y_1y_2}}-H(\Abar_2|\Abar_1E_1E_2\Bbar_1)_{\zeta^{x_1x_2y_1y_2}}\right],
\end{align}
where
\begin{align}
\zeta_{\Abar_1\Abar_2E_1E_2}^{x_1,x_2,y_1,y_2}&=  \sum_{a_1}t(a_1|x_1)\left[a_1\right]_{\Abar_1}\otimes \rho_{E_1}^{a_1,x_1}\otimes \sum_{a_2}r(a_2|x_2)\left[a_2\right]_{\Abar_2}\otimes\rho_{E_2}^{a_2,x_2},\label{add_1}\\
\zeta_{\Abar_1\Abar_2\Bbar_1E_1E_2}^{x_1,x_2,y_1,y_2}&= \sum_{a_1,b_1}t(a_1,b_1|x_1,y_1)\left[a_1b_1\right]_{\Abar_1\Bbar_1}\otimes \rho_{E_1}^{a_1,x_1,b_1,y_1}\otimes \sum_{a_2}r(a_2|x_2)\left[a_2\right]_{\Abar_2}\otimes\rho_{E_2}^{a_2,x_2}.\label{add_2}
\end{align}
Then, from \eqref{add_1} and \eqref{add_2}, it follows that
\begin{align}
H(\Abar_2|\Abar_1 E_1E_2)_{\zeta^{x_1x_2y_2y_2}}&=H(\Abar_2|E_2)_{\zeta^{x_1x_2y_2y_2}} , \\
H(\Abar_2|\Abar_1 E_1E_2\Bbar_1)_{\zeta^{x_1x_2y_2y_2}}&=H(\Abar_2|E_2)_{\zeta^{x_1x_2y_2y_2}}.
\end{align}
This is equivalent to $I(\Abar_2;\Bbar_1|E_1E_2X_1Y_1X_2Y_2\Abar_1)_{\zeta}=0$.

Similarly, $I(\Abar_1;\Bbar_2|E_1E_2X_1Y_1X_2Y_2\Bbar_1)_{\zeta}=0$. 
Then by inspection of \eqref{A2}, and from the no-signaling constraints, it follows that
\begin{align}
\inf_{\text{ext. in }\eqref{A1}}I(\Abar_1\Abar_2;\Bbar_1\Bbar_2|X_1Y_1X_2Y_2E)_{\rho}\leq I(\Abar_1;\Bbar_1|X_1Y_1E_1)_\zeta+ I(\Abar_2;\Bbar_2|X_2Y_2E_2)_{\zeta}.
\end{align}
Since the above statement holds for an arbitrary no-signaling extension of the form in \eqref{A1}, it follows that
\begin{multline}
\inf_{\text{ext. in \eqref{A1}}}I(\Abar_1\Abar_2;\Bbar_1\Bbar_2|X_1Y_1X_2Y_2E)_{\rho}\\\leq \inf_{\text{ext. in \eqref{A2}}}I(\Abar_1;\Bbar_1|X_1Y_1E_1)_\zeta+ \inf_{\text{ext. in \eqref{A2}}}I(\Abar_2;\Bbar_2|X_2Y_2E_2)_{\zeta}.
\end{multline}
Since the above inequality holds for an arbitrary probability distribution $p(x_1,x_2,y_1,y_2)$, we find that
\begin{multline}
\sup_{p(x_1,x_2,y_1,y_2)}\inf_{\text{ext. in  (\ref{A1})}}I(\Abar_1\Abar_2;\Bbar_1\Bbar_2|X_1Y_1X_2Y_2E)_{\rho}\\ \leq \sup_{p(x_1,y_1)}\inf_{\text{ext. in (\ref{A2})}}I(\Abar_1;\Bbar_1|X_1Y_1E_1)_\zeta+ \sup_{p(x_2,y_2)}\inf_{\text{ext. in (\ref{A2})}}I(\Abar_2;\Bbar_2|X_2Y_2E_2)_{\zeta} .
\end{multline}
This concludes the proof. 
\end{proof}

\begin{proposition}[Superadditivity and additivity of quantum intrinsic non-locality]
\label{prop:additivity-quantum}
Let \\
$p(a_1,a_2,b_1,b_2|x_1,x_2,y_1,y_2)$ be a quantum correlation that arises from a four-party state $\rho_{A_1A_2B_1B_2}$ and POVMs characterized by $\Lambda^{a_1}_{x_1},\Lambda^{a_2}_{x_2},\Lambda^{b_1}_{y_1}$, and $\Lambda^{b_2}_{y_2}$.  Then the following no-signaling constraints hold: 
\begin{align}
&\sum_{a_1}p(a_1,a_2,b_1,b_2|x_1,x_2,y_1,y_2)=\sum_{a_1}p(a_1,a_2,b_1,b_2|x_1',x_2,y_1,y_2) \quad \forall x_1',x_1 ,x_2,y_1,y_2\in [s],\,  a_2,b_1,b_2 \in [r],\nonumber\\
&\sum_{a_2}p(a_1,a_2,b_1,b_2|x_1,x_2,y_1,y_2)=\sum_{a_2}p(a_1,a_2,b_1,b_2|x_1,x_2',y_1,y_2) \quad \forall x_2',x_2 ,x_1,y_1,y_2\in [s],\,  a_1,b_1,b_2 \in [r],\nonumber \\
&\sum_{b_1}p(a_1,a_2,b_1,b_2|x_1,x_2,y_1,y_2)=\sum_{b_1}p(a_1,a_2,b_1,b_2|x_1,x_2,y_1',y_2) \quad \forall y_1',y_1 ,x_1,x_2,y_2\in [s],\,  a_1,a_2,b_2 \in [r], \nonumber\\
&\sum_{b_2}p(a_1,a_2,b_1,b_2|x_1,x_2,y_1,y_2)=\sum_{b_2}p(a_1,a_2,b_1,b_2|x_1,x_2,y_1,y_2') \quad \forall y_2',y_2 ,x_2,y_1,x_1\in [s],\,  a_1,a_2,b_1 \in [r]. \nonumber
\end{align}
Let $t(a_1,b_1|x_1,y_1)$ and $r(a_2,b_2|x_2,y_2)$ be quantum correlations corresponding to the marginals of $p(a_1,a_2,b_1,b_2|x_1,x_2,y_1,y_2)$. Then the quantum intrinsic non-locality is super-additive, in the sense that
\begin{equation}\label{eq:quantum-superadditivity}
N^Q(\Abar_1\Abar_2;\Bbar_1\Bbar_2)_{p}\geq N^Q(\Abar_1;\Bbar_1)_t+N^Q(\Abar_2;\Bbar_2)_r.
\end{equation}
If $p(a_1,b_1,a_2,b_2|x_1,x_2,y_1,y_2)=t(a_1,b_1|x_1,y_1) r(a_2,b_2|x_2,y_2)$, then the quantum intrinsic non-locality is additive in the following sense:
\begin{equation}
N^Q(\Abar_1\Abar_2;\Bbar_1\Bbar_2)_p= N^Q(\Abar_1;\Bbar_1)_t+N^Q(\Abar_2;\Bbar_2)_r.
\end{equation}
\end{proposition}
\begin{proof}
The proof follows by using similar techniques as Proposition~\ref{prop:additivity}, and by taking appropriate quantum extensions. 
\end{proof}

\bigskip
Let $\rho_{AB}$ be quantum state, and let $p_{\bar{A}|X}(a|x)\rho_B^{a,x}$ be an assemblage that arises from the quantum state $\rho_{AB}$ and some measurement $\left\{\Lambda_a^x\right\}$.\footnote{From \cite{post-quantum-steering}, it can be seen that given a bipartite assemblage, we can always find an underlying quantum state and measurements.} We then prove that the intrinsic steerability of the assemblage $p_{\bar{A}|X}(a|x)\rho_B^{a,x}$ is never smaller than the quantum intrinsic non-locality of all the bipartite correlations that can arise from this assemblage. 
\begin{proposition}\label{prop:inequality}
Let $p(a,b|x,y)$ be a quantum correlation that is obtained by performing a POVM $\left\{\Lambda_y^b\right\}_b$ on the assemblage $\{p_{\bar{A}|X}(a|x)\rho_B^{a,x}\}_{a,x}$. Then the quantum intrinsic non-locality of the correlation $p$ does not exceed the intrinsic steerability of the assemblage $\hat{\rho}$. That is,
\begin{equation}
N^Q(\bar{A};\bar{B})_{p} \leq S(\bar{A};B)_{\hat{\rho}},
\end{equation}
where we recall that $\hat{\rho}$ is a shorthand to denote the assemblage.
\end{proposition}

\begin{proof}
Let $p(a,b|x,y)$ be a quantum correlation that arises from the assemblage $p_{\bar{A}|X}(a|x)\rho_B^{a,x}$. That is, 
\begin{equation}
p(a,b|x,y)= \operatorname{Tr}\left[\Lambda^b_y\left(p_{\bar{A}|X}(a|x)\rho_B^{a,x}\right)\right].
\end{equation}
Let $p_{\bar{A}|X}(a|x)\rho_{BE}^{a,x}$ be a particular no-signaling extension of $p_{\bar{A}|X}(a|x)\rho_B^{a,x}$. Then one possible no-signaling extension of $p(a,b|x,y)$ is 
\begin{equation}
p(a,b|x,y)\rho_E^{a,x,b,y}= \operatorname{Tr}_B\left[\Lambda^b_y\left(p_{\bar{A}|X}(a|x)\rho_{BE}^{a,x}\right)\right]
\end{equation}
From \cite{post-quantum-steering}, it follows that the above is also a quantum extension. 

Let $p(x,y)$ be an arbitrary probability distribution. Let $p(a,b|x,y)$ be a correlation embedded in a classical-classical state $\rho_{\bar{A}\bar{B}XY}$ with the following particular no-signaling extension:
\begin{equation}\label{ext:correlation}
 \rho_{\bar{A}\bar{B}XYE}:=\sum_{a,b,x,y} p(x,y)p(a,b|x,y) \left[a\,b\, x\,y\right]_{\bar{A}\bar{B}XY} \otimes \rho_E^{a,b,x,y},
\end{equation}
and an arbitrary quantum extension: 
\begin{equation}\label{eqn:correlationp}
 \sigma_{\bar{A}\bar{B}XYE}:=\sum_{a,b,x,y} p(x,y)p(a,b|x,y) \left[a\,b\, x\,y\right]_{\bar{A}\bar{B}XY} \otimes \sigma_E^{a,b,x,y}.
\end{equation}
Similarly, let $\rho_{\bar{A}XB}$ be  a state into which the assemblage $p_{\Abar|X}(a|x)\rho_{B}^{a,x}$ is embedded, and let $\rho_{\bar{A}XBE}$ be a particular extension, where
\begin{equation}\label{ext:assemblages}
    \rho_{\bar{A}BXE}= \sum_{a,x} p(x)p_{\Abar|X}(a|x)[a\,x]_{\bar{A}X}\otimes \rho_{BE}^{a,x}.
\end{equation}
Let 
\begin{equation}\label{eqref:add_systems}
    \rho_{\bar{A}BXYE}= \sum_{a,x} p(x,y)p_{\Abar|X}(a|x)[a\,x]_{\bar{A}X}\otimes \rho_{BE}^{a,x}.
\end{equation}
Then, 
\begin{equation}
I(\Abar;B|XE)_{\rho}= I(\Abar;BY|XE)_{\rho}
\end{equation}
This follows from the chain rule of conditional mutual information and inspection of \eqref{eqref:add_systems}. 
Observe that Bob can perform a local operation and transform the state $\rho_{\bar{A}BXYE}$ to $\rho_{\bar{A}\bar{B}XYE}$. Then, from the data-processing inequality, we find that
\begin{equation}\label{cmi_inequality}
    I(\bar{A};B|XE)_{\rho}\geq I(\bar{A};\bar{B}Y|XE)_{\rho}.
\end{equation}
This means that for every no-signaling extension $\rho_{\bar{A}BXE}$ of the state $\rho_{\bar{A}BX}$ that encodes the assemblage $p_{\Abar|X}(a|x)\rho_B^{a,x}$, we can find a quantum extension $\rho_{\Abar\Bbar XYE}$ of $\rho_{\Abar\Bbar XY}$ that encodes the correlation $p(a,b|x,y)$ derived from the assemblage $p_{\Abar|X}(a|x)\rho_B^{a,x}$, such that \eqref{cmi_inequality} is true. Therefore, we obtain the following: 
\begin{align}
    \inf_{\mathrm{ext\, in\,} \eqref{ext:assemblages}}I(\bar{A};B|XE)_{\rho}&\geq \inf_{\mathrm{ext.\,in\,} \eqref{ext:correlation}}I(\bar{A};\bar{B}Y|XE)_{\rho}\\&\geq\inf_{\mathrm{ext\, in\,} \eqref{eqn:correlationp}}I(\bar{A};\bar{B}Y|XE)_{\sigma}.
\end{align}
This in turn implies that
\begin{equation}
    S(\bar{A};B)_{\hat{\rho}}\geq N^Q(\bar{A};\bar{B})_p,
 \end{equation}
 concluding the proof. 
\end{proof}

\subsection{Intrinsic non-locality of a PR box}

In this section, we calculate the intrinsic non-locality of a PR box. 

\begin{proposition}
The intrinsic non-locality of a PR box is equal to $1$, i.e., $N(\Abar;\Bbar)_p =1$, where~$p$ is the correlation defined in \eqref{E4}.  
\end{proposition}
\begin{proof}
Consider the state
\begin{equation}
\rho_{\Abar\Bbar X Y}:=\sum_{a,b,x,y}p(x,y)p(a,b|x,y)\[a\,b\,x\,y\]_{\Abar\Bbar XY},
\end{equation}
where $p(x,y)$ is an arbitrary probability distribution. 
Consider a no-signaling extension of the state 
\begin{equation}
\rho_{\Abar\Bbar X Y E}:=\sum_{a,b,x,y}p(x,y)p(a,b|x,y)\[a\,b\,x\,y\]_{\Abar\Bbar XY}\otimes \rho_E^{a,b,x,y}. \label{E1}
\end{equation}
The no-signaling constraints are 
\begin{align}
\sum_{a,b,y} p(a,b|x,y)\[b\,x\,y\]_{\Bbar XY}\otimes \rho_E^{x,y,a,b}&=\sum_{a,b,y} p(a,b|x',y)\[b\,x'\,y\]_{\Bbar XY}\otimes \rho_E^{x',y,a,b},\label{NS1}\\
\sum_{b,a,x} p(a,b|x,y)\[a\,x\,y\]_{\Abar XY}\otimes \rho_E^{x,y,a,b}&=\sum_{b,a,x} p(a,b|x,y')\[a\,x\,y'\]_{\Abar XY}\otimes \rho_E^{x,y',a,b}.\label{NS2}
\end{align}
From \eqref{E4}, and the no-signaling constraint in \eqref{NS1}, we arrive at the following constraints on the possible states of Eve's system:
\begin{equation}\label{eq:no-sig-cons-prbox}
\begin{bmatrix}
   \rho_E^{0000}      & 0 & 0  & 0 \\
    0       & \rho_E^{0011} & 0 &  0 \\
    0       & 0  & \rho_E^{0100} &0 \\
    0       & 0 & 0 &\rho_E^{0111}
\end{bmatrix}
=
\begin{bmatrix}
    \rho_E^{1000}      & 0 & 0  & 0 \\
    0       & \rho_E^{1011} & 0 &  0 \\
    0       & 0  & \rho_E^{1110} &0 \\
    0       & 0 & 0 &\rho_E^{1101}
\end{bmatrix}.
\end{equation}
In the matrices given above, the rows and columns are indexed by $(y,b)$. The first matrix on the left corresponds to $x=0$, and the second one  on the right corresponds to $x=1$. The constraints in (\ref{eq:no-sig-cons-prbox}) can also be written as 
\begin{align}
1)\; \;\rho_E^{0000}&=\rho_E^{1000},\quad  &2)\;\;\rho_E^{0011}&=\rho_E^{1011},\nonumber\\
3)\;\;\rho_E^{0100}&=\rho_E^{1110}, \quad &4)\;\;\rho_E^{0111}&=\rho_E^{1101}. 
\end{align}
Similarly, from \eqref{E4}, and the no-signaling constraint in \eqref{NS2}, we arrive at the following constraints on the possible states of Eve's system:
\begin{equation}\label{eq:no-sig-cons-prbox2}
\begin{bmatrix}
   \rho_E^{0000}      & 0 & 0  & 0 \\
    0       & \rho_E^{0011} & 0 &  0 \\
    0       & 0  & \rho_E^{1000} &0 \\
    0       & 0 & 0 &\rho_E^{1011}
\end{bmatrix}
=
\begin{bmatrix}
    \rho_E^{0100}      & 0 & 0  & 0 \\
    0       & \rho_E^{0111} & 0 &  0 \\
    0       & 0  & \rho_E^{1101} &0 \\
    0       & 0 & 0 &\rho_E^{1110}
\end{bmatrix}.
\end{equation}
In the above block matrices, the rows and columns are indexed by $(x,a)$. The first matrix on the left corresponds to $y=0$, and the second one on the right corresponds to $y=1$. The constraints in (\ref{eq:no-sig-cons-prbox2}) can also be written as
\begin{align}
5)\; \;\rho_E^{0000}&=\rho_E^{0100},\quad  &6)\;\;\rho_E^{0011}&=\rho_E^{0111},\nonumber\\
7)\;\;\rho_E^{1000}&=\rho_E^{1101}, \quad &8)\;\;\rho_E^{0111}&=\rho_E^{1101}. 
\end{align}
By following $1 \rightarrow 7 \rightarrow 4 \rightarrow 6 \rightarrow 2 \rightarrow 8 \rightarrow 3 \rightarrow 5 \rightarrow 1$ in the above, we obtain $\rho_E^{x,y,a,b} = \rho_E^{x',y',a',b'} \quad \forall x,x',y,y' \in [s]$ and $a,a',b,b' \in [r]$. This implies that $\rho_{\bar{A}\bar{B}XY}$ has a trivial tensor product no-signaling extension. Hence, 
\begin{align}
I(\Abar;\Bbar|XYE)_{\rho}=I(\Abar;\Bbar|XY)_{\rho}&=\sum_{x,y}p(x,y)I(\Abar;\Bbar)_{\rho^{x,y}}\\&=\sum_{x,y}p(x,y)\left(H(\Abar)_{\rho^{x,y}}-H(\Abar|\Bbar)_{\rho^{x,y}}\right)\\&=1.
\end{align}
It is easy to check that given realizations of $X,Y$, the entropies $H(\Abar|\Bbar)_{\rho^{x,y}}=0$ and $H(\Abar)_{\rho^{x,y}}=1$.
\end{proof}

\section{Faithfulness of restricted intrinsic steerability} \label{sec:faithfulness_steerability}

In this section, we solve an open question from \cite{Kaur2016}, regarding the faithfulness of restricted intrinsic steerability.

\begin{theorem}[Faithfulness of restricted intrinsic steerability]
For every assemblage $\hat{\rho}_{B}^{a,x}$, the restricted intrinsic steerability $S(A;B)_{\hat{\rho}}=0$ if and only if it is an LHS assemblage. Quantitatively, if $S(\bar{A};B)_{\hat{\rho}}\leq \varepsilon$, where  $0<\varepsilon^{\frac{1}{16}}|\mathcal{X}|^{\frac{1}{2}}< 1$, there exists an LHS assemblage $\sigma_{\bar{A}XB}$ such that 
\begin{equation}
\sup_{p_X(x)}\left\|\rho_{\bar{A}XB}-\sigma_{\bar{A}XB}\right\|_1\leq |\mathcal{X}| \left(\varepsilon^{1/4}+\frac{\varepsilon^{1/16}|\mathcal{X}|^{1/2}}{1-\varepsilon^{1/16}|\mathcal{X}|^{1/2}}+4|\mathcal{X}|e^{-\frac{\varepsilon^{-1/4}}{3}}\right).
\end{equation} 

\end{theorem}

\begin{proof}
The forward direction (``if'') was established in \cite[Proposition~12]{Kaur2016}. We now give a proof for the reverse direction (``only if'') of the theorem. 

Let us first construct a proof strategy for a uniform probability distribution $p_X(x)=\frac{1}{|\mathcal{X}|}$, and then we generalize it to a proof for an arbitrary distribution $p_X(x)$. This proof shares some ideas from the proof for faithfulness of squashed entanglement \cite{Li}. 

Invoking Theorem 5.1 of \cite{ Fawzi2015}, we know that there exists a recovery channel $\mathcal{R}_{XE\rightarrow \bar{A}XE}$ such that
\begin{align}
\left\|\rho_{\bar{A}XBE}-\mathcal{R}_{XE\rightarrow \bar{A}XE}(\rho_{BE}\otimes\rho_X)\right\|_1&\leq \sqrt{I(\bar{A};B|EX)_{\rho} \ln2}=:t,\label{S1}\\
\left\|\rho_{\bar{A}XBE}-\mathcal{R}_{X_2E\rightarrow {\bar{A}_2}X_2E}{\circ}\operatorname{Tr}_{\bar{A}_1X_1}(\rho_{\bar{A}_1X_1BE}\otimes \rho_{X_2})\right\|_1&\leq  t,\label{Sa}
\end{align}
where systems $\bar{A}_1$ and $\bar{A}_2$ are isomorphic to system $\bar{A}$, and systems $X_1$ and $X_2$ are isomorphic to~$X$. In the above, we have invoked the no-signaling condition $I(X;BE)_{\rho}=0$, which implies that $\rho_{BE}$ and $\rho_X$ are product as written. Now, let us apply this recovery channel again. We then have that
\begin{multline}\label{Sb}
\left\|\mathcal {R}_{X_3E\rightarrow \bar{A}_3X_3E}\circ\operatorname{Tr}_{X_2\bar{A}_2} (\rho_{\bar{A}_2X_2BE}\otimes \rho_{X_3})-\bigcirc_{i=2}^{3}\mathcal {R}_{X_iE\rightarrow \bar{A}_i X_iE}\circ\operatorname{Tr}_{A_{i-1}X_{i-1}}(\rho_{\bar{A}_{1}X_{1}BE}\otimes \rho_{X_2}\otimes \rho_{X_3})\right\|_1\\\leq  t.
\end{multline}
which follows from the monotonicity of trace distance with respect to $\mathcal {R}_{X_3E\rightarrow \bar{A}_3X_3E}\circ\operatorname{Tr}_{X_2\bar{A}_2} $. 
Then, combining the above equation with \eqref{S1} via the triangle inequality, we obtain 

\begin{equation}\label{Sc}
\left\|\rho_{\bar{A}XBE}-\bigcirc_{i=2}^{3}\mathcal {R}_{X_iE\rightarrow \bar{A}_i X_iE}\circ\operatorname{Tr}_{A_{i-1}X_{i-1}}(\rho_{\bar{A}_{1}X_{1}BE}\otimes \rho_{X_2}\otimes \rho_{X_3})\right\|_1\leq 2 t.
\end{equation}
For $j\in \{4,\ldots, n\}$, again apply the channels $\mathcal{R}_{XE\rightarrow \bar{A}_jX_jE}\circ\operatorname{Tr}_{\bar{A}_{j-1}X_{j-1}}$, along with the monotonicity of trace norm under quantum channels, combining the equations via the triangle inequality, to obtain the following inequality: 
\begin{equation} 
\left\|\rho_{\bar{A}XB}-\operatorname{Tr}_E\{\bigcirc_{i=2}^{j}\mathcal{R}_{X_iE\rightarrow \bar{A}_i X_iE}\circ\operatorname{Tr}_{A_{i-1}X_{i-1}}\left(\rho_{\bar{A}_{1}X_{1}BE}\otimes \rho_{X}^{\otimes j}\right)\}\right\|_1 \leq n t.\label{Sd}
\end{equation}
The recovery channel $\mathcal{R}_{X_iE\rightarrow \bar{A}_i X_iE}$ can be taken as \cite{W15}
\begin{align}
\mathcal{R}_{XE\rightarrow \bar{A} X E}\left(\cdot\right)&=\rho_{\Abar XE}^{\frac{1}{2}+i\omega}\rho_{XE}^{-\frac{1}{2}-i\omega}\left(\cdot\right)\rho_{XE}^{-\frac{1}{2}+i\omega}\rho_{\Abar XE}^{\frac{1}{2}-i\omega},\\
&= \sum_x \op{x}_X \otimes (\rho_{\Abar E}^x)^{\frac{1}{2}+i\omega}\rho_E^{-\frac{1}{2}+i\omega}\left(\cdot\right)\rho_E^{-\frac{1}{2}+i\omega}(\rho_{\Abar E}^x)^{\frac{1}{2}-i\omega},
\end{align}
for some $\omega \in \mathbb{R}$.
Let $\sigma_{\Abar^nX^nBE}$ denote the following state:
\begin{align}
\sigma_{\bar{A}^nX^nBE}&=\left(\mathcal{R}_{X_nE\rightarrow\bar{A}_nX_nE}\circ\cdots\circ \mathcal{R}_{X_1E\rightarrow \bar{A}_1X_1E}\right)(\sigma_{BE}\otimes \sigma_X^{\otimes n})\\
&=\sum_{a^n,x^n}p_{X^n}(x^n)q_{\bar{A}^n|X^n}(a^n|x^n)\op{x^n}{x^n}_{X^n}\otimes\op{a^n}{a^n}_{\bar{A}^n}\otimes \sigma_{BE}^{a^n,x^n}.\\
\sigma_{\bar{A}^nX^nB}&=\operatorname{Tr}_E(\sigma_{\bar{A}^nX^nBE})\\
&=\sum_{a^n,x^n}p_{X^n}(x^n)q_{\bar{A}^n|X^n}(a^n|x^n) \op{x^n}{x^n}_{X^n}\otimes\op{a^n}{a^n}_{\bar{A}^n}\otimes\sigma_{B}^{a^n,x^n}.\\
\sigma_{\bar{A}_iX_iB}&=\operatorname{Tr}_{A^{\[n\]\backslash\left\{i\right\}}X^{\[n\]\backslash\left\{i\right\}}}(\sigma_{\bar{A}^nX^nB})\\
&=\sum_{a^n,x^n}p_{X^n}(x^n)q_{\bar{A}^n|X^n}(a^n|x^n)\op{x_i}{x_i}_{X_i}\otimes\op{a_i}{a_i}_{\bar{A}_i}\otimes\sigma_B^{a^n,x^n}, \label{F1}
\end{align}
where $A^{\[n\]\backslash\left\{i\right\}}= A_1A_2\cdots A_{i-1}A_{i+1}\cdots A_n$ and similarly $X^{\[n\]\backslash\left\{i\right\}}= X_1X_2\cdots X_{i-1}X_{i+1}\cdots X_n$. Furthermore, $q_{\Abar^n|X^n}(a^n|x^n)$ is a probability distribution for $a^n$ given $x^n$ after the application of the recovery channels $\mathcal{R}_{X_iE\rightarrow \Abar_iX_iE}$.
From \eqref{Sd}, we obtain for all $i \in \left\{1,2,\ldots, n\right\}$ that
\begin{equation} \label{Se}
\left\|\rho_{\bar{A}XB}-\sigma_{\bar{A}_iX_iB}\right\|_1\leq nt.
\end{equation}

The application of the recovery channels generates the data $(x_1,a_1),(x_2,a_2),\ldots,(x_n,a_n)$. The $x_i$ correspond to the measurement choices, and the $a_i$ correspond to the measurement outcomes. This data is called the ``cheat sheet'' and acts like a hidden variable $\lambda$. The formulation of the cheat sheet is similar to the construction of a local hidden-variable model in \cite{Terhal2002}. 
\par We now devise an algorithm to generate $\at$ from $\xx$ by using the cheat sheet. The generated state $\sigma_{\tilde{A}\tilde{X}B}$ is a local hidden state, with the cheat sheet as the hidden variable.  We then prove that $\sigma_{\tilde{A}\tilde{X}B}$ is close to the original state $\rho_{\bar{A}XB}$.

\par
Alice receives $\xx$. She searches for all the values of $i$ for which $x_i=\xx$, and generates $i$ uniformly at random
\begin{equation}
p_{I|\tilde{X}X^n}(i|\tilde{x}x^n)=\frac{1}{N(\tilde{x}|x^n)}\delta_{x_i\tilde{x}},
\end{equation}
where $\delta_{x_i\tilde{x}}$ is the Kronecker delta function and where $N(\xx|x^n)$ is the number of times that the letter $\xx$ appears in the sequence $x^n$.
Then, she outputs $\at$ with probability 
\begin{equation}
p_{\tilde{A}|A^n I}(\tilde{a}|a^ni)= \delta_{\tilde{a},a_i}.
\end{equation} 
Therefore,
\begin{align}
p_{\tilde{A}|\tilde{X}X^nA^n}(\tilde{a}|\tilde{x}x^na^n)&= \sum_{i=1}^n p_{\tilde{A}|A^nIX^n\tilde{X}}(\tilde{a}|a^nix^n\tilde{x})p_{I|\tilde{X}X^nA^n}(i|\tilde{x}x^na^n)\\
&= \sum_{i=1}^n p_{\tilde{A}|A^nI}(\tilde{a}|a^n i)p_{I|\tilde{X}X^n}(i|\tilde{x}x^n).\\
&=\sum_{i=1}^n \frac{1}{N(\tilde{x}|x^n)}\delta_{\tilde{x}x_i}\delta_{\tilde{a}a_i}.
\end{align}
If $\xx$ does not belong to the sequence $x^n$, then she generates $\at$ randomly. This sequence of actions can be expressed in terms of the following conditional probability distribution:
\begin{equation}\label{eqn:algorithm}
p_{\tilde{A}|\tilde{X}X^nA^n}(\at|\xx,x^n,a^n):= \begin{cases} \frac{1}{|A|}, \quad\text{if } N(\xx|x^n)=0 \\  \sum_{i=1}^n \frac{1}{N(\xx|x^n)}\delta_{\xx, x_i}\delta_{a_i,\at} \quad \text{else}. \end{cases}
\end{equation}
 It is easy to check that $\sum_{\tilde{a}}p_{\tilde{A}|\tilde{X}X^nA^n}(\at|\xx,x^n,a^n)=1.$
\par We now use the notion of robust typicality \cite{Roche2001} for the analysis. 
\begin{definition}[Robust typicality \cite{Roche2001}]
\label{def:strongly}
Let $x^n$ be a sequence of elements drawn from a finite alphabet $\mathcal{X}$, and let $p(x)$ be a probability distribution on $\mathcal{X}$. Let $N(x|x^n)$ be the empirical distribution of $x^n$. Then the $\delta$-robustly typical set $T^{X^n}_{\delta}$  for $\delta>0$ is defined as 
\begin{align}
T^{X^n}_{\delta}:=\left\{x^n:\forall \mathcal{X},\left|\frac{1}{n}N(x|x^n)-p_X(x)\right|\leq \delta p(x)
\right\}.
\end{align}

\end{definition}
The following result holds for $0<\delta<1$:
\begin{Property}\label{prop:typical}
The probability for a sequence $x^n$ to be in the robustly typical set is bounded from below as 
\begin{equation}\label{eq:robust_typ}
\mathrm{Pr}\left\{X^n\in T^{X^n}_\delta\right\}\geq 1-2|\mathcal{X}|\exp^{-\frac{n\delta^2\mu_X}{3}},
\end{equation}
where 
\begin{equation}
\mu_X := \min_{x\in\mathcal{X},p_X(x)> 0} p_X(x).
\end{equation}
\end{Property}
The state generated after the application of the algorithm in \eqref{eqn:algorithm} is as follows:
\begin{equation}
\sigma_{\tilde{A}\tilde{X}B}= \sum_{\tilde{x},\tilde{a}}p_{\tilde{X}}(\tilde{x})\op{\tilde{x}}_{\tilde{X}}\otimes \sum_{x^n,a^n}p_{\tilde{A}|\tilde{X}X^nA^n}(\tilde{a}|\tilde{x},x^n,a^n)p_{X^n}(x^n)q_{\bar{A}^n|X^n}(a^n|x^n)\op{\tilde{a}}_{\tilde{A}}\otimes \sigma_B^{a^n,x^n}.
\end{equation}
Then, define the following sets:
\begin{itemize}
\item $S_1(\tilde{x})$: set of sequences $x^n$ such that $\tilde{x}\in x^n$ and $x^n\in T^{X^n}_{\delta}$,
\item $S_2(\tilde{x})$: set of sequences $x^n$ such that $\tilde{x}\not\in x^n$ and $ x^n\in T^{X^n}_{\delta}$,
\item $S_3$: set of sequences $x^n$ such that $x^n\not\in T^{X^n}_{\delta}.$ 
\end{itemize}
So we can write the state $\sigma_{\tilde{A}\tilde{X}B}$ as
\begin{align}
\sigma_{\tilde{A}\tilde{X}B}= &\sum_{\tilde{x},\tilde{a}}p_{\tilde{X}}(\tilde{x})\op{\tilde{x}}_{\tilde{X}}\otimes\bigg( \sum_{x^n \in S_1(
\tilde{x}),a^n} p(\tilde{a}|\tilde{x},x^n,a^n)\op{\tilde{a}}\otimes q(a^n,x^n) \sigma_B^{a^n,x^n}+\nonumber\\ &\sum_{x^n \in S_2(
\tilde{x}),a^n} p(\tilde{a}|\tilde{x},x^n,a^n)\op{\tilde{a}}\otimes q(a^n,x^n)\sigma_B^{a^n,x^n}+ \sum_{x^n \in S_3,a^n} p(\tilde{a}|\tilde{x},x^n,a^n)\op{\tilde{a}}\otimes q(a^n,x^n) \sigma_B^{a^n,x^n}\bigg),\\
\sigma_{\tilde{A}\tilde{X}B}& = \sigma_{\tilde{A}\tilde{X}B}^{(1)}+\sigma_{\tilde{A}\tilde{X}B}^{(2)}+\sigma_{\tilde{A}\tilde{X}B}^{(3)}.
\end{align}
From the triangle inequality, we obtain the following:
\begin{align}
\left\|\rho_{\bar{A}\bar{X}B}-\sigma_{\tilde{A}\tilde{X}B}\right\|_1\leq \left\|\rho_{\bar{A}\bar{X}B}-\sigma_{\tilde{A}\tilde{X}B}^{(1)}\right\|_1 +\left\| \sigma_{\tilde{A}\tilde{X}B}^{(2)}\right\|_1+\left\| \sigma_{\tilde{A}\tilde{X}B}^{(3)}\right\|_1, 
\end{align}
where
\begin{align}
\left\|\rho_{\bar{A}XB}-\sigma^{(1)}_{\tilde{A}\tilde{X}B}\right\|_1 &\leq \left\|\rho_{\bar{A}\bar{X}B}-\frac{1}{n}\sum_{i=1}^{n}\sigma_{\bar{A}_iX_iB}\right\|_1 +\left\|\frac{1}{n}\sum_{i=1}^{n}\sigma_{\bar{A}_iX_iB}-\sigma^{(1)}_{\tilde{A}\tilde{X}B}\right\|_1\\
&\leq nt + \left\|\frac{1}{n}\sum_{i=1}^{n}\sigma_{\bar{A}_iX_iB}-\sigma^{(1)}_{\tilde{A}\tilde{X}B}\right\|_1 \label{S3}. 
\end{align}
Let us analyze each term individually, beginning with
\begin{align}
\left\|
\sigma_{\tilde{A}\tilde{X}B}^{(3)}\right\|_1&= \left\|\sum_{\tilde{x},\tilde{a}}p_{\tilde{X}}(\tilde{x})\op{\tilde{x}}_{\tilde{X}}\otimes\sum_{x^n \in S_3,a^n} p(\tilde{a}|\tilde{x},x^n,a^n)\op{\tilde{a}}\otimes q(a^n,x^n) \sigma_B^{a^n,x^n}\right\|_1\\
&\leq \sum_{\tilde{x},\tilde{a}}p(\tilde{x})\sum_{x^n\in S_3,a^n}p(x^n) q(a^n|x^n)p(\tilde{a}|\tilde{x},x^n,a^n)\left\|\op{\tilde{x}}\otimes\op{\tilde{a}}\otimes\sigma_B^{a^n,x^n}\right\|_1\\
&= \sum_{\tilde{x}}p(\tilde{x})\sum_{x^n\in S_3}p(x^n)\sum_{a^n}q(a^n|x^n)\sum_{\tilde{a}}p(\tilde{a}|\tilde{x},x^n,a^n)\leq \varepsilon_1,\label{Faith2}
\end{align}
where $\varepsilon_1 =2|\mathcal{X}|\exp^{-\frac{n\delta^2\mu_X}{3}}$.
The first inequality follows from convexity of trace distance, and the second inequality follows from the definition of $S_3$ and \eqref{eq:robust_typ}.

Let us now consider $S_2(\tilde{x})$, that is, the set of sequences $x^n$ such that $\tilde{x}\not\in x^n$ and $x^n \in T^{X^n}_{\delta}$. 
From Definition \ref{def:strongly}, we know that for the robustly-typical set, the following condition holds
\begin{align}
x^n:\forall x\in\mathcal{X} ,\quad\left|\frac{1}{n}N(x|x^n)-p_X(x)\right|\leq \delta p_X(x).\label{eqn:robust_def}
\end{align}
For a robustly-typical sequence to have an empirical distribution $N(x|x^n)=0$, it is required that $\delta\geq 1$. So, we restrict $\delta \in (0,1)$. Thus, by the fact that $p_X(x)> 0$ for all $ x \in\mathcal{X}$, it is impossible for $N(\tilde{x}|x^n)=0$ and $x^n\in T^{X^n}_{\delta}$. That is, 
\begin{equation}\label{Faith3}
\left\|\sigma^{(2)}_{\tilde{A}\tilde{X}B}\right\|_1=0.
\end{equation}
Consider that
\begin{align}
&\sigma_{\tilde{X}\tilde{A}B}^{(1)}\nonumber \\&= \sum_{\tilde{x}}p(\tilde{x})[\tilde{x}]_{\tilde{X}}\otimes \sum_{a^n,x^n\in S_1(\tilde{x}),\tilde{a}}\sum_{i=1}^n \frac{1}{N(\tilde{x}|x^n)}\delta_{a_i,\tilde{a}} \delta_{\tilde{x},x_i}[\tilde{a}]_{\tilde{A}} \otimes p_{X^n}(x^n)q_{A^n|X^n}(a^n|x^n)\sigma_B^{a^n,x^n},\\
&=\sum_{\tilde{x}}p(\tilde{x})[\tilde{x}]_{\tilde{X}}\otimes\sum_{\tilde{a}}\[\tilde{a}\]_{\tilde{A}}\otimes\frac{1}{n}\sum_{i=1}^n\sum_{x^{\[n\]\backslash\left\{i\right\},\tilde{x}}\in S_1(\tilde{x}),a^{\[n\]\backslash\left\{i\right\}}}\frac{p_{\tilde{X}}(\tilde{x})}{N(\tilde{x}|x^n)/n}\,p_{X^{\[n\]\backslash\left\{i\right\}}}(x^{\[n\]\backslash\left\{i\right\}}|\tilde{x})q(\tilde{a}|x^{\[n\]\backslash\left\{i\right\},\tilde{x}})\nonumber\\&\quad q(a^{\[n\]\backslash\left\{i\right\}}|x^{\[n\]\backslash\left\{i\right\},\tilde{x}}\tilde{a})\sigma_B^{a^{\[n\]\backslash\left\{i\right\}},x^{\[n\]\backslash\left\{i\right\}},\tilde{x},\tilde{a}},
\end{align}
where $x^{\[n\]\backslash\left\{i\right\},\tilde{x}}$ refers to a sequence $x^n$ with $x_i = \tilde{x}$.

We now want to give an upper bound on the second term in \eqref{S3}: 
\begin{align}
&\left\|\frac{1}{n}\sum_{i=1}^{n}\sigma_{\bar{A}_iX_iB}-\sigma^{(1)}_{\bar{A}XB}\right\|_1,\label{eqn:faithfulness_trace}
\end{align}
where
\begin{align}
\sigma_{\Abar_iX_iB}=\sum_{a^n,x^n}p_{X^n}(x^n)q_{\bar{A}^n|X^n}(a^n|x^n)\op{x_i}{x_i}_{X_i}\otimes\op{a_i}{a_i}_{\bar{A}_i}\otimes\sigma_B^{a^n,x^n}.
\end{align}
Let us define the following sets:
\begin{itemize}
\item $S_1(x_i)$: set of sequences $x^n$ such that $x_i\in x^n $ and $ x^n\in T^{X^n}_{\delta}$,
\item $S_2(x_i)$: set of sequences $x^n$ such that $x_i\not\in x^n$ and $ x^n\in T^{X^n}_{\delta}$,
\item $S_3$: set of sequences $x^n$ such that $x^n\not\in T^{X^n}_{\delta}.$ 
\end{itemize}
Then, 
\begin{align}
\sigma_{\Abar_iX_iB}=&\sum_{a^n,x^n\in S_1(x_i)}p_{X^n}(x^n)q_{\bar{A}^n|X^n}(a^n|x^n)\op{x_i}{x_i}_{X_i}\otimes\op{a_i}{a_i}_{\bar{A}_i}\otimes\sigma_B^{a^n,x^n}\nonumber
\\
&\qquad +\sum_{a^n,x^n\in S_2(x_i)}p_{X^n}(x^n)q_{\bar{A}^n|X^n}(a^n|x^n)\op{x_i}{x_i}_{X_i}\otimes\op{a_i}{a_i}_{\bar{A}_i}\otimes\sigma_B^{a^n,x^n}\nonumber
\\
&\qquad +\sum_{a^n,x^n\in S_3}p_{X^n}(x^n)q_{\bar{A}^n|X^n}(a^n|x^n)\op{x_i}{x_i}_{X_i}\otimes\op{a_i}{a_i}_{\bar{A}_i}\otimes\sigma_B^{a^n,x^n}\\
&= \sigma_{\Abar_iX_iB}^{(1)}+\sigma_{\Abar_iX_iB}^{(2)}+\sigma_{\Abar_iX_iB}^{(3)}.
\end{align}
Then, using the convexity of trace distance with \eqref{eqn:faithfulness_trace} and typicality arguments similar to \eqref{Faith2} and \eqref{Faith3}, we find that
\begin{align}
\left\|\frac{1}{n}\sum_{i=1}^{n}\sigma_{\bar{A}_iX_iB}-\sigma^{(1)}_{\bar{A}XB}\right\|_1&\leq\frac{1}{n}\sum_{i=1}^{n}\left\|\sigma_{\bar{A}_iX_iB}-\sigma^{(1),i}_{\bar{A}XB}\right\|_1\\&\leq \frac{1}{n}\sum_{i=1}^{n}\left\|\sigma_{\bar{A}_iX_iB}^{(1)}-\sigma^{(1),i}_{\bar{A}XB}\right\|_1+\varepsilon_1,\label{eqn:final1}
\end{align}
where
\begin{multline}
\sigma_{\Abar_iX_iB}^{(1)}=\sum_{x_i}p_{X_i}(x_i)\[x_i\]_{X_i}\otimes \sum_{a_i}\[a_i\]_{\Abar_i}
\\
\otimes\sum_{x^{\[n\]\backslash\left\{i\right\},x_i}\in S_1(x_i),a^{\[n\]\backslash\left\{i\right\}}}p(x^{\[n\]\backslash\left\{i\right\}}|x_i)q(\tilde{a}|x^{\[n\]\backslash\left\{i\right\}},x_i)q(a^{\[n\]\backslash\left\{i\right\}}|x^{\[n\]\backslash\left\{i\right\}},x_i,\tilde{a})\sigma_B^{a^{\[n\]\backslash\left\{i\right\}}x^{\[n\]\backslash\left\{i\right\}},x_i,a_i}.
\end{multline}
and 
\begin{multline}
\sigma^{(1),i}_{\bar{A}XB}=\sum_{\tilde{x}}p(\tilde{x})[\tilde{x}]_{\tilde{X}}\otimes\sum_{\tilde{a}}\[\tilde{a}\]_{\tilde{A}}\otimes\sum_{x^{\[n\]\backslash\left\{i\right\},\tilde{x}}\in S_1(\tilde{x}),a^{\[n\]\backslash\left\{i\right\}}}\frac{p_{\tilde{X}}(\tilde{x})}{N(\tilde{x}|x^n)/n}\,p_{X^{\[n\]\backslash\left\{i\right\}}}(x^{\[n\]\backslash\left\{i\right\}}|\tilde{x})q(\tilde{a}|x^{\[n\]\backslash\left\{i\right\},\tilde{x}})\\\quad q(a^{\[n\]\backslash\left\{i\right\}}|x^{\[n\]\backslash\left\{i\right\},\tilde{x}}\tilde{a})\sigma_B^{a^{\[n\]\backslash\left\{i\right\}},x^{\[n\]\backslash\left\{i\right\}},\tilde{x},\tilde{a}}.
\end{multline}
Invoking \eqref{eqn:robust_def}, we find that
\begin{equation}
\frac{1}{n}\sum_{i=1}^{n}\left\|\sigma_{\bar{A}_iX_iB}^{(1)}-\sigma^{(1),i}_{\bar{A}XB}\right\|_1\leq\frac{\delta}{1-\delta},\label{eqn:final2}
\end{equation}
where $\delta \in (0,1)$.
After combining \eqref{Faith2}, \eqref{Faith3}, \eqref{eqn:final1}, and \eqref{eqn:final2}, we obtain
\begin{equation}
\left\|\rho_{\bar{A}XB}-\sigma_{\tilde{A}\tilde{X}B}\right\|_1\leq n t+\frac{\delta}{1-\delta}+2\varepsilon_1.
\end{equation}
Minimizing over all possible no-signaling extensions, as required by the definition, we find that
\begin{equation}
\left\|\rho_{\bar{A}XB}-\sigma_{\tilde{A}\tilde{X}B}\right\|_1\leq n \inf_{\rho_{\bar{A}XBE}}t+\frac{\delta}{1-\delta}+2\varepsilon_1.
\end{equation}
Since $\rho_{\bar{A}XB}$ and $\sigma_{\bar{A}XB}$ are classical-quantum states with $p_X(x)=\frac{1}{|\mathcal{X}|}$, we obtain
\begin{equation}
\sum_{x}\left\|\rho_{\bar{A}B}^x-\sigma_{\tilde{A}B}^x\right\|_1\leq |\mathcal{X}| \left(n \inf_{\rho_{\bar{A}XBE}}t+\frac{\delta}{1-\delta}+2\varepsilon_1\right).
\end{equation}
This implies that the following inequality holds for all $x\in \mathcal{X}$:
\begin{equation}
\left\|\rho_{\bar{A}B}^x-\sigma_{\tilde{A}B}^x\right\|_1\leq |\mathcal{X}| \left(n \inf_{\rho_{\bar{A}XBE}}t+\frac{\delta}{1-\delta}+2\varepsilon_1\right).
\end{equation}
This means that we can average the above to get a bound for any arbitrary distribution $p(x)$ on~$x$. Therefore, we can now relax the assumption of a uniform probability distribution, in order to obtain the following bound for an arbitrary probability distribution:
\begin{equation}
\sup_{p_X(x)}\left\|\rho_{\bar{A}BX}-\sigma_{\tilde{A}BX}\right\|_1\leq |\mathcal{X}| \left(n \sup_{p_X(x)}\inf_{\rho_{\bar{A}XBE}}t+\frac{\delta}{1-\delta}+2\varepsilon_1\right),
\end{equation}
which implies that\begin{equation}
\sup_{p_X(x)}\left\|\rho_{\bar{A}BX}-\sigma_{\tilde{A}BX}\right\|_1\leq |\mathcal{X}| \left(n \sqrt{ S(\bar{A};B)_{\hat{\rho}}\ln2}+\frac{\delta}{1-\delta}+2\varepsilon_1\right).
\end{equation}
Given $S(\bar{A};B)_{\hat{\rho}} \leq \varepsilon$ (as required by the condition of faithfulness), choose $n = (1/\varepsilon)^{1/4}$, $\delta = \varepsilon^{1/16}|\mathcal{X}|^{1/2}$ (recall that we require $\delta \in (0,1)$). We know by the Chernoff bound \cite{Roche2001} that $\varepsilon_1 = 2 |\mathcal{X}|e^{-\frac{1}{3|\mathcal{X}|}\delta^2 n}$. Substituting these values, we find that
\begin{equation}
\left\|\rho_{\bar{A}BX}-\sigma_{\tilde{A}BX}\right\|_1\leq |\mathcal{X}| \left(\varepsilon^{1/4}+\frac{\varepsilon^{1/16}|\mathcal{X}|^{1/2}}{1-\varepsilon^{1/16}|\mathcal{X}|^{1/2}}+4|\mathcal{X}|e^{-\frac{\varepsilon^{-1/4}}{3}}\right).
\end{equation}
This concludes the proof. 
\end{proof}

\section{Faithfulness of intrinsic non-locality}\label{sec:faithfulness_nonlocality}

The following theorem, combined with Proposition \ref{prop:prof_faithful}, establishes that intrinsic non-locality is faithful. 

\begin{theorem}[Faithfulness of intrinsic non-locality] \label{theorem:faithfulness_nonlocality}
For every no-signaling or quantum correlation $p(a,b|x,y)$, the intrinsic non-locality $N(\bar{A};\bar{B})_p=0$, if and only if it has a local hidden variable description. Quantitatively, if $N(\bar{A};\bar{B})_{p}\leq \varepsilon$, where $0<\varepsilon^{1/16}d^{1/2}< 1$, for $d = |\mathcal{X}|\cdot |\mathcal{Y}|$, there exists a probability distribution $l(a,b|x,y)$ having a local hidden variable description, such that 
\begin{equation}
\sup_{p_{XY}(x,y)}\left\|\rho_{\bar{A}X\bar{B}Y}-\gamma_{\bar{A}X\bar{B}Y}\right\|_1\leq d\left(\varepsilon^{1/4}+\frac{\varepsilon^{1/16}d^{1/2}}{1-\varepsilon^{1/16}d^{1/2}}+4d e^{-\frac{\varepsilon^{-1/4}}{3}}\right),
\end{equation}
where $\rho_{\bar{A}X\bar{B}Y}$ correponds to the classical-classical state $p_{XY}(x,y)p(a,b|x,y)$ and $\gamma_{\bar{A}X\bar{B}Y}$ is the classical-classical state corresponding to $p_{XY}(x,y)l(a,b|x,y)$.
\end{theorem}
\begin{proof}
The proof closely follows the proof for faithfulness of intrinsic steerability. We first construct a strategy for $p_{XY}(x,y)= \frac{1}{|\mathcal{X}|}.\frac{1}{|\mathcal{Y}|}$ and then generalize it to an arbitrary distribution. Invoking \cite{Fawzi2015}, we know that there exists a recovery channel $\mathcal{R}_{XE\rightarrow \bar{A}XE}$ such that 
\begin{equation}
\|\rho_{\bar{A}X\bar{B}YE}-\mathcal{R}_{XE\rightarrow \bar{A}XE}(\rho_{\bar{B}YE}\otimes \rho_X)\|_1\leq \sqrt{I(\bar{A};\bar{B}Y|XE)_\rho \ln 2}=t.
\end{equation}
Since $I(\bar{B}E;X|Y)_{\rho}=0$ from \eqref{eqn:no_signaling_constraint}, and $p_{XY}(x,y)= \frac{1}{\mathcal{X}}.\frac{1}{\mathcal{Y}}$, we can write $\rho_{\Bbar XYE}= \rho_{\Bbar YE}\otimes \rho_X$.  
Following an argument similar to \eqref{Sa}--\eqref{Sd}, we obtain the following inequality:
\begin{equation}
\|\rho_{\bar{A}\bar{B}XY}- \omega_{A_iX_iBY}\| \leq nt,
\end{equation}
where 
\begin{align}
\omega_{\bar{A}^nX^n\bar{B}YE} &= \bigcirc_{i=1}^{n}\mathcal{R}_{X_iE\rightarrow \bar{A}_iX_iE}\(\rho_{\bar{B}YE}\otimes \rho^{\otimes n}_{X}\),\\
\omega_{\bar{A}_iX_iBY}&= \operatorname{Tr}_{E\bar{A}^{n/i}X^{n/i}}\left(\omega_{\bar{A}^nX^n\bar{B}YE}\right).
\end{align}
Since the distributions $p_X(x)$ and $p_Y(y)$ are independent, we have 
\begin{align}
I(X^n;Y)_{p}=0.
\end{align}
From the no-signaling constraints, we have 
\begin{align}
I(X^nY;E)_{\rho}=0.
\end{align}
This implies that
\begin{equation}
I(X^nE;Y)_{\rho}=I(X^n;Y)_{\rho}+I(E;Y|X^n)_{\rho}=0.
\end{equation}
Since the systems $\Abar^n X^n E$ of $\omega_{\Abar^n X^n\Bbar YE}$ are obtained from the application of the recovery channel on systems $X_nE$ of the state $\rho_{X_nYE\Bbar}$, we can use quantum data processing for mutual information to obtain the following inequality:
\begin{equation}
I(A^nX^n;Y)_{\omega}=0. 
\end{equation}
This implies that
\begin{equation}
\omega_{\bar{A}^nX^n\bar{B}Y} = \sum_{x^n,a^n,y,b} p(x^n)\, q(a^n|x^n) \, p(y)\, q(b|a^nx^ny) \[x^n\,a^n\,b\,y\]_{X^n\Abar^n\Bbar Y}.
\end{equation}
Alice's strategy is exactly the same as before, and the following state is obtained after the application of the algorithm in \eqref{eqn:algorithm}:
\begin{multline}
\gamma_{\tilde{A}\tilde{X}\bar{B}Y}:= \sum_{\tilde{x},\tilde{a},b,y}p_X(\tilde{x})\op{\tilde{x}}{\tilde{x}}_{\tilde{X}}\otimes \sum_{x^n,a^n,}p_{\tilde{A}|\tilde{X}X^nA^n}(\tilde{a}|\tilde{x},x^n,a^n)p_{X^n}(x^n)\\q_{A^n|X^n}(a^n|x^n)p(y)q(b|a^nx^ny)[\tilde{a}\,b\,y]_{\tilde{A}\bar{B}Y}.
\end{multline}
Note that this state is a local hidden-variable state. This construction of the local hidden-variable state shares some similarities with \cite{Terhal2002}. By following the arguments given for the proof of faithfulness of intrinsic steerability, we obtain 
\begin{equation}
\|\rho_{\bar{A}X\bar{B}Y}-\gamma_{\tilde{A}\tilde{X}\bar{B}Y}\|_1\leq  nt+\frac{\delta}{1-\delta}+2\varepsilon_1. 
\end{equation}
This implies 
\begin{equation}
\|\rho_{\bar{A}X\bar{B}Y}-\gamma_{\tilde{A}\tilde{X}\bar{B}Y}\|_1\leq  n\inf_{\rho_{\bar{A}X\bar{B}YE}}t+\frac{\delta}{1-\delta}+2\varepsilon_1. 
\end{equation}
This implies
\begin{equation}
\sum_{a,b}|p(a,b|x,y)-l(a,b|x,y)| \leq |\mathcal{X}| |\mathcal{Y}|\left(\inf_{\rho_{\bar{A}X\bar{B}YE}}t+\frac{\delta}{1-\delta}+2\varepsilon_1\right) \quad \forall x\in \mathcal{X}\, , y\in \mathcal{Y}. 
\end{equation}
Now, using triangle inequality, we obtain the following for any arbitrary distribution $p(x,y)$:
\begin{equation}
\|\rho_{\bar{A}X\bar{B}Y}-\gamma_{\tilde{A}\tilde{X}\bar{B}Y}\|_1 \leq |\mathcal{X}||\mathcal{Y}|\left(\inf_{\rho_{\bar{A}X\bar{B}YE}}t+\frac{\delta}{1-\delta}+2\varepsilon_1\right). 
\end{equation}
This implies 
\begin{equation}
\sup_{p_{XY}(x,y)}\|\rho_{\bar{A}X\bar{B}Y}-\gamma_{\tilde{A}\tilde{X}\bar{B}Y}\|_1 \leq |\mathcal{X}||\mathcal{Y}|\left(\sqrt{N(\bar{A};\bar{B})_p\ln 2}+\frac{\delta}{1-\delta}+2\varepsilon_1\right).
\end{equation}
Given $N(\bar{A};\bar{B})_p \leq \varepsilon$ (as required by the condition of faithfulness), choose $n = (1/\varepsilon)^{1/4}$, $\delta = \varepsilon^{1/16}|\mathcal{X}|^{1/2}|\mathcal{Y}|^{1/2}$. This proof holds only if $\delta \in (0,1)$. We know by the Chernoff bound \cite{Roche2001} that $\varepsilon_1 = 2 |\mathcal{X}||\mathcal{Y}|e^{-\frac{1}{3|\mathcal{X}|\cdot |\mathcal{Y}|}\delta^2 n}$. Substituting these values, we obtain
\begin{equation}
\left\|\rho_{\bar{A}X\bar{B}Y}-\gamma_{\tilde{A}\tilde{X}\bar{B}Y}\right\|_1\leq |\mathcal{X}|\cdot|\mathcal{Y}|\left(\varepsilon^{1/4}+\frac{\varepsilon^{1/16}|\mathcal{X}|^{1/2}\cdot|\mathcal{Y}|^{1/2}}{1+\varepsilon^{1/16}|\mathcal{X}|^{1/2}\cdot |\mathcal{Y}|^{1/2}}+4|\mathcal{X}|\cdot|\mathcal{Y}|e^{-\frac{\varepsilon^{-1/4}}{3}}\right).
\end{equation}
This concludes the proof. 
\end{proof}

\begin{corollary}[Faithfulness of quantum intrinsic non-locality]
For every quantum correlation $p(a,b|x,y)$, the quantum intrinsic non-locality $N^Q(\bar{A};\bar{B})_p=0$, if and only if it has a local hidden-variable description. Quantitatively, if $N^Q(\bar{A};\bar{B})_{p}\leq \varepsilon$, where $0<\varepsilon^{1/16}d^{1/2}< 1$, for $d = |\mathcal{X}|\cdot |\mathcal{Y}|$, there exists a probability distribution $l(a,b|x,y)$ having a local hidden-variable description, such that 
\begin{equation}
\sup_{p_{XY}(x,y)}\left\|\rho_{\bar{A}X\bar{B}Y}-\gamma_{\bar{A}X\bar{B}Y}\right\|_1\leq d\left(\varepsilon^{1/4}+\frac{\varepsilon^{1/16}d^{1/2}}{1-\varepsilon^{1/16}d^{1/2}}+4d e^{-\frac{\varepsilon^{-1/4}}{3}}\right),
\end{equation}
where $\rho_{\bar{A}X\bar{B}Y}$ correponds to the classical-classical state $p_{XY}(x,y)p(a,b|x,y)$ and $\gamma_{\bar{A}X\bar{B}Y}$ is the classical-classical state corresponding to $p_{XY}(x,y)l(a,b|x,y)$.
\end{corollary}
\begin{proof}
The if-part of the proof follows from Proposition~\ref{prop:prof_faithful}. The only-if part follows from Proposition~\ref{prop:relation} and Theorem~\ref{theorem:faithfulness_nonlocality}. 
\end{proof}

\section{Upper bounds on secret key rates in device-independent quantum key distribution}

\label{section:upper_bounds}

We now consider the task of device-independent quantum key distribution. We consider two honest parties, Alice and Bob, who share a two-component device and want to extract a shared secret key from this device.

In general, in the device-independent literature, many prior works have devised lower bounds on the key rates for particular protocols, as done in \cite{Acin2007, Rotem2018}. By a protocol, we mean a sequence of steps in which Alice and Bob interact with their devices and communicate publicly with each other.

Here, we are interested in a different question. We fix the black-box device that is shared by Alice and Bob. We suppose that the correlations  generated from this device are characterized by a correlation $p(a,b|x,y)$. We then pose the following question: 
\begin{quote}
  Given a device characterized by $p(a,b|x,y)$, what is a non-trivial upper bound on the secret-key rate that can be extracted from this device with any possible protocol?  
\end{quote}
 We answer this question for an i.i.d.~device, which means that in each round of the protocol, the device considered is characterized by the correlation $p(a,b|x,y)$. The inputs of the device in a particular round can be correlated with the inputs of the device in other rounds. The assumption that the device is characterized by the correlation $p(a,b|x,y)$ is not a drawback since we are interested in determining upper bounds on secret-key rates here. In what follows, we prove that the quantifiers introduced above are upper bounds on the secret-key rates that can be generated from the device. 

In device-independent quantum key distribution, we assume the presence of an eavesdropper who obtains all of the classical data communicated between Alice and Bob during the protocol. Furthermore, the system held by the eavesdropper can have joint correlations with the systems held by Alice and Bob. Let Alice and Bob share a quantum correlation $p(a,b|x,y)$ as defined in \eqref{eqn:quantum-correlation}. Let the correlation shared between Alice, Bob, and Eve be defined by $p(a,b|x,y)\rho_E^{a,b,x,y}$. If $p(a,b|x,y)\rho_E^{a,b,x,y}$ has an underlying quantum strategy as described in \eqref{constraint_3}, then we call the eavesdropper a quantum Eve. If $p(a,b|x,y)\rho_E^{a,b,x,y}$ only fulfills the constraints given in \eqref{eq:constraint_1} and \eqref{eq:constraint_2}, then we call the eavesdropper a no-signaling Eve.

\subsection{Device-independent protocols}

We now state the general form of a device-independent protocol with no-signaling eavesdropper for which our upper bounds hold. Such protocols have previously been considered in \cite{Barrett2005,Masanes2009,Masanes2014}. Let $n \in \mathbb{Z}^+$, $R \geq 0$, and $\varepsilon \in [0,1]$. Let $p(a,b|x,y)$ be the correlation of the device shared between Alice and Bob. We define an $(n,R,\varepsilon)$ device-independent secret-key-agreement protocol as follows:
\begin{itemize}
\item Alice and Bob give the inputs $x^n$ and $y^n$ to their devices according to $p_{X^nY^n}(x^n,y^n)$. The device is used $n$ times, and the distribution $p_{X^nY^n}(x^n,y^n)$ is independent of Eve. Alice inputs $x_i$ and obtains the output $a_i$. Bob inputs $y_i$ and obtains the output $b_i$, where $i\in \{1,\ldots, n\}$. The input and output distributions are embedded in the state $\sigma_{\Abar^n\Bbar^nX^nY^n}$, where
\begin{equation}\label{eqn:protocol-state}
    \sigma_{\Abar^n\Bbar^nX^nY^n}
    := \sum_{x^n, y^n, a^n, b^n} p_{X^nY^n}(x^n,y^n) p^n(a^n,b^n|x^n, y^n) [a^n b^n x^n y^n]_{\Abar^n \Bbar^n X^n Y^n},
\end{equation}
and $p^n(a^n,b^n|x^n, y^n)$ is the i.i.d.~extension of  $p(a,b|x,y)$.
The joint state held by Alice, Bob, and Eve is a no-signaling extension $\sigma_{\bar{A}^n\bar{B}^nX^nY^nE}$ of $\sigma_{\bar{A}^n\bar{B}^nX^nY^n}$. 

\item Alice and Bob perform local operations and public communication, with $C_A$ denoting the classical register communicated from Alice to Bob, $\bar{C}_A$ is a classical register held by Eve that is a copy of $C_A$, the classical register $C_B$ is  communicated from Bob to Alice, and $\bar{C}_B$ is a classical register held by Eve that is a copy of  $C_B$. This protocol yields a state $\omega_{K_AK_BE\bar{C}_A\bar{C}_BX^nY^n}$ that satisfies 
\begin{equation}
\left\|\omega_{K_AK_BEX^nY^n\bar{C}_A\bar{C}_B}-\overline{\Phi}_{K_AK_B}\otimes \omega_{EX^nY^n\bar{C}_A\bar{C}_B}\right\|_1 \leq \varepsilon,
\end{equation}
for all no-signaling extensions,
where 
\begin{equation}
\overline{\Phi}_{K_AK_B} = \frac{1}{2^{nR}}\sum_{k=1}^{2^{nR}}\op{kk}_{K_AK_B}.
\end{equation}
\end{itemize}

A rate $R$ is achievable for a device characterized by $p$ if there exists an $(n,R-\delta,\varepsilon)$ device-independent protocol for all $\varepsilon \in (0,1)$, $\delta > 0$, and sufficiently large $n$. The device-independent secret-key-agreement capacity $DI(p)$ of the device characterized by $p$ is defined to be equal to the supremum of all achievable rates.

\begin{theorem}\label{theorem:device_independent}
The intrinsic non-locality $N(\bar{A};\bar{B})_p$ is an upper bound on the device-independent secret-key-agreement capacity of a device characterized by $p$ and sharing no-signaling correlations with an eavesdropper:
\begin{equation}
    DI(p) \leq N(\bar{A};\bar{B})_p.
\end{equation}
\end{theorem}

\begin{proof}
For an arbitrary $(n,R,\varepsilon)$ protocol, consider that
\begin{align}
nR&=I(K_A;K_B|EX^nY^n\bar{C}_A\bar{C}_B)_{\bar{\Phi}\otimes\omega}\\ 
&\leq I(K_A;K_B|EX^nY^n\bar{C}_A\bar{C}_B)_{\omega}+\varepsilon'\\
&\leq I(M_AC_BC_A;M_BC_BC_A|EX^nY^n\bar{C}_A\bar{C}_B)_{\tau}+\varepsilon'\\
& = I(M_AC_A;M_BC_B|EX^nY^n\bar{C}_A\bar{C}_B)_{\tau}+\varepsilon'\\
&\leq I(M_AC_A;M_BC_B|EX^nY^n)_{\tau}+\varepsilon'\\
&\leq I(\bar{A}^n;\bar{B}^n|EX^nY^n)_{\sigma}+\varepsilon', \label{eq:device_independent}
\end{align}
where 
\begin{equation}
\varepsilon'= nR\varepsilon+ 2 \[(1+\varepsilon) \log (1+\varepsilon))-\varepsilon\log \varepsilon\].
\end{equation}
In the above equations, $\sigma_{X^n\Abar^n\Bbar^n Y^n}$ is the classical-classical state obtained from the device after Alice and Bob enter in the measurement inputs. Alice, Bob, and Eve hold a no-signaling extension $\sigma_{X^n\Abar^n\Bbar^n Y^nE}$. Alice performs a local operation $\mathcal{L}_A$ to obtain $M_A$ and $C_A$. She communicates $C_A$ to Bob, and Eve also obtains a copy $\bar{C}_A$ of the classical communication. Similarly, Bob performs a local operation $\mathcal{L}_B$ to obtain $M_B$ and $C_B$. He communicates $C_B$ to Alice, and Eve also obtains a copy $\bar{C}_B$ of the classical communication. Alice then performs a local operation $\mathcal{D}_A$ on $M_A$, $C_B$, and $C_A$ to obtain $K_A$, while Bob performs a local operation $\mathcal{D}_B$ on $M_B$, $C_A$, and $C_B$ to obtain $K_B$. For a pictorial representation of the above description, refer to Figure~\ref{fig:protocol_device_independent}. 
\begin{figure}
    \centering
    \includegraphics[width=4in]{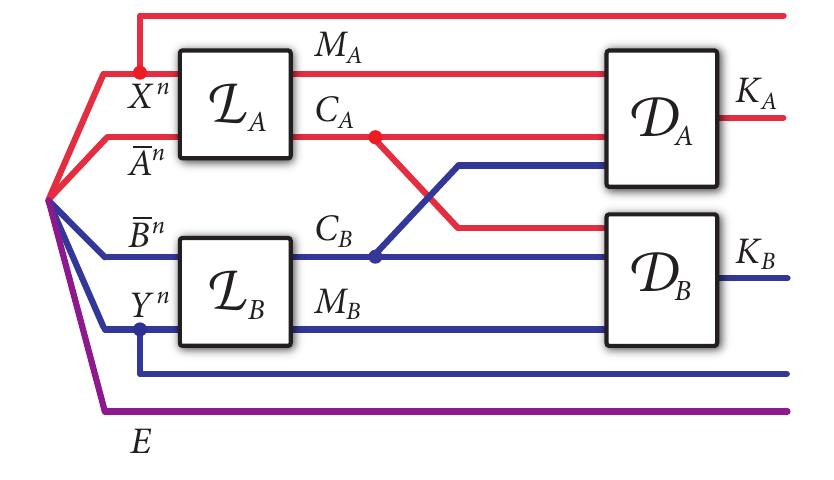}
    \caption{This figure depicts a device-independent secret-key agreement protocol, along with the various states involved at each time step. The classical systems $\bar{A}^nX^n \bar{B}^nY^n$ are processed by local operations and classical communication to produce classical key systems $K_A$ and $K_B$ that are approximately independent of the no-signaling extension system $E$, the input systems $X^n$ and $Y^n$, and copies $\bar{C}_A$ and $\bar{C}_B$ of the classical systems $C_A$ and $C_B$, respectively.}
    \label{fig:protocol_device_independent}
\end{figure}
\par The first inequality follows from the uniform continuity of conditional mutual information \cite[Proposition~1]{Shirokov2017}. The second inequality follows from data processing. The second equality and third inequality follow from the chain rule of conditional mutual information, as well as the fact that $\bar{C}_A$ is a classical copy of $C_A$ and $\bar{C}_B$ is a classical copy of $C_B$. The last inequality follows from data processing for conditional mutual information.
Since the above inequality holds for an arbitrary no-signaling extension of $\sigma_{\Abar^n\Bbar^nX^nY^n}$, we find that 
\begin{equation}
nR\leq \inf_{\sigma_{\Abar^n\Bbar^nX^nY^nE}}I(\Abar^n;\Bbar^n|X^nY^nE)_{\sigma} + \varepsilon'.
\end{equation}
This implies that
\begin{equation}
nR\leq N(\Abar^n;\Bbar^n)_{p}+ \varepsilon'.
\end{equation}

By the assumption that the device is i.i.d.,  we can invoke the additivity of intrinsic non-locality from Proposition~\ref{prop:additivity} to obtain
\begin{equation}
(1-\varepsilon)R \leq N(\Abar;\Bbar)_{p} + 2 \[(1+\varepsilon) \log (1+\varepsilon))-\varepsilon\log \varepsilon\]/n.
\end{equation}
Taking the limit as $n \rightarrow \infty$ and $\varepsilon \rightarrow 0$ then leads to $DI(p)\leq N(\Abar;\Bbar)_p$. 
\end{proof}

\bigskip

Now, let us consider a class of device-independent protocols in which the eavesdropper is restricted by quantum mechanics. These models have previously been studied in \cite{Acin2007,Rotem2018}. The general form of a device-independent protocol with a quantum eavesdropper remains the same except that we now consider a quantum extension \eqref{constraint_3} of the state in \eqref{eqn:protocol-state}. We then arrive at the following theorem: 

\begin{theorem}\label{theorem:quantum_device_independent}
The quantum intrinsic non-locality $N^Q(\bar{A};\bar{B})_p$ is an upper bound on the device-independent secret-key-agreement capacity of a device characterized by $p$ and sharing quantum correlations with an eavesdropper:
\begin{equation}
    DI(p) \leq N^Q(\bar{A};\bar{B})_p.
\end{equation}
\end{theorem}
\begin{proof}
The proof of the theorem is similar to that of Theorem~\ref{theorem:device_independent}.
\end{proof}

\bigskip
We should explicitly point out that the general form for protocols that we consider allow both Alice and Bob to exchange public classical information. Therefore, the upper bounds via intrinsic non-locality and quantum intrinsic non-locality hold for two-way error correction as well. It has been observed in device-dependent QKD that two-way error-correcting protocols surpass the threshold of one-way error-correcting protocols \cite{JA2007,Watanbe2007,Khatri2017}. This question has only recently been explored in DI-QKD in \cite{Ernest2019}. Therefore, it is possible that the upper bound via the intrinsic non-locality will not be tight for the existing DI-QKD protocols \cite{Acin2007,Rotem2018} which consider only one-way error correction.  

Another point to make is that in the protocols we consider, Alice and Bob announce their measurement choices. That is, $X$ and $Y$ are known to Eve. The secret key is extracted from $\Abar$ and $\Bbar$. There are certain protocols in the device-independent literature where the outputs $
\Abar$ and $\Bbar$ are broadcast and the local randomness variables $X$ and $Y$ are the basis of the key \cite{RPMP15} (note that \cite{SARG04} introduced this concept in the device-dependent QKD literature). For such DI-QKD protocols, our upper bounds do not hold. 

\subsubsection{Other considerations}

Bounds on device-independent QKD protocols based on certain states were also previously discussed in \cite{HM2015}.

There is yet another way to model a no-signaling adversary in the device-independent secret agreement protocols which has been considered in \cite{Barrett2005}. This model is set in ``box world,'' in which each player including the eavesdropper has a set of possible inputs and outputs. Therefore, it becomes natural to model the joint system with a conditional probability distribution $P_{ABE|XYZ}$. In \cite{WDH19}, the authors introduced squashed non-locality to provide an upper bound on key rates of device-independent protocols with the aforementioned model of the eavesdropper. This is in contrast to the model that we consider where the eavesdropper is a quantum no-signaling adversary but is not equipped with a number of measurements.

\subsection{One-sided-device-independent protocol}

Let  $n \in \mathbb{Z}^+$, $R \geq 0$, and $\varepsilon \in [0,1]$. We define an $(n,R,\varepsilon)$ one-sided-device-independent secret-key-agreement protocol for an assemblage $\hat{\rho}:=\{p_{A|X}(a|x)\rho^{a,x}_B\}_{a,x}$ as follows:
\begin{itemize}
\item Alice gives input $x^n$ to get an output $a^n$. The assemblage shared by Alice and Bob is then
\begin{equation}
\rho_{\Abar^nX^nB^n}:= \sum_{x^n,a^n}p_{X^n}(x^n)p_{A^n|X^n}(a^n|x^n)\[x^n,a^n\]_{X^nA^n}\otimes\rho_{B^n}^{a^n,x^n},
\end{equation}
where $\{p_{A^n|X^n}(a^n|x^n)\rho_{B^n}^{a^n,x^n}\}_{a^n,x^n}$ is an i.i.d.~extension of the assemblage $\{p_{A|X}(a|x)\rho^{a,x}_B\}_{a,x}$.
Alice, Bob, and Eve hold a no-signaling extension of the above assemblage: 
\begin{equation}
\rho_{\Abar^nX^nB^nE}:=
\sum_{x^n,a^n}p_{X^n}(x^n)p_{A^n|X^n}(a^n|x^n)\[x^n,a^n\]_{X^nA^n}\otimes\rho_{B^nE}^{a^n,x^n}.
\end{equation}
\item  Bob inputs $y_i$ and obtains the output $b_i$, where $i\in \left\{1, \ldots, n\right\}$. Let the measurement corresponding to $y^n$ be a set $\left\{Y_{b^n}^n\right\}_{b^n}$ of measurement operators, such that $\sum_{b^n} (Y_{b^n}^n)^{\dagger}Y_{b^n}^n=I$.  The state shared between Alice, Bob and Eve is then $\sigma_{\bar{A}^nX^n\bar{B}^nY^nE}$.
\begin{multline}
\sigma_{\Abar^nX^nY^n\bar{B}^nE}:= \sum_{x^n,a^n}p_{X^n}(x^n)p_{\Abar^n|X^n}(a^n|x^n)\[x^n,a^n\]_{X^n\Abar^n}\otimes \sum_{y^nb^n}p_{Y^n}(y^n)[y^n]_{Y_n}\otimes \\(Y^n_{b^n}\rho_{B^nE}^{a^n,x^n}(Y^n_{b^n})^\dagger).
\end{multline}
\item Alice and Bob perform local operations and public communication, with $C_A$ being the classical register communicated from Alice to Bob, $\bar{C}_A$ is a classical register held by Eve that is a copy of $C_A$,  the classical register $C_B$ is communicated from Bob to Alice, and $\bar{C}_B$ is a classical register held by Eve that is a copy of  $C_B$. This protocol yields a state $\omega_{K_AK_BE\bar{C}_A\bar{C}_BX^nY^n}$ that satisfies 
\begin{equation}
\left\|\omega_{K_AK_BEX^nY^n\bar{C}_A\bar{C}_B}-\overline{\Phi}_{K_AK_B}\otimes \omega_{EX^nY^n\bar{C}_A\bar{C}_B}\right\|_1 \leq \varepsilon,
\end{equation}
for all no-signaling extensions, where
\begin{equation}
\overline{\Phi}_{K_AK_B} = \frac{1}{2^{nR}}\sum_{k=1}^{2^{nR}}\op{kk}_{K_AK_B}.
\end{equation}
\end{itemize}

A rate $R$ is achievable for a device characterized by $\hat{\rho}$ if there exists an $(n,R-\delta,\varepsilon)$ one-sided device-independent protocol for all $\varepsilon \in (0,1)$, $\delta > 0$, and sufficiently large $n$. The one-sided device-independent capacity $SDI(\hat{\rho})$ of the device characterized by $\hat{\rho}$ is defined to be equal to the supremum of all achievable rates for $\hat{\rho}$. 

\begin{theorem}
The restricted intrinsic steerability $S(\Abar;\Bbar)_{\hat{\rho}}$ is an upper bound on the one-sided device-independent secret-key-agreement capacity $SDI(\hat{\rho})$ of a device characterized by $\hat{\rho}$:
\begin{equation}
    SDI(\hat{\rho}) \leq S(\bar{A};B)_{\hat{\rho}}.
\end{equation}
\end{theorem}
\begin{proof}
For obtaining an upper bound in the one-sided device-independent setting, we continue from $\eqref{eq:device_independent}$ as follows:
\begin{align}
nR &\leq I(\Abar^n;\Bbar^nY^n|EX^n)_{\sigma}-I(\Abar^n;Y^n|EX^n)_{\sigma}+ \varepsilon'\\
&\leq I(\Abar^n;\Bbar^nY^n|EX^n)_{\sigma}+\varepsilon'\\
&\leq I(\Abar^n;B^n|EX^n)_{\rho}+\varepsilon',
\end{align}
The first inequality follows from the chain rule of conditional mutual information. The last inequality follows from data processing. 
Since the above inequality holds for an arbitrary no-signaling extension of $\rho_{\Abar^nX^nB^n}$, we obtain
\begin{equation}
nR\leq \inf_{\rho_{\Abar^nX^nB^nE}}I(\Abar^n;B^n|X^nE)_{\rho}+\varepsilon'.
\end{equation}
This implies that
\begin{equation}
nR\leq S(\bar{A}^n;B^n)_{\hat{\rho}} +\varepsilon'.
\end{equation}
Since we assume an i.i.d.~device, we find by applying the additivity of restricted intrinsic steerability \cite{Kaur2016} that
\begin{equation}
(1-\varepsilon)R \leq S(\Abar;B)_{\hat{\rho}} + 2 \[(1+\varepsilon) \log (1+\varepsilon))-\varepsilon\log \varepsilon\]/n.
\end{equation}
Taking the limit as $n \rightarrow \infty$ and $\varepsilon \rightarrow 0$ then leads to the desired inequality $SDI(\hat{\rho})\leq S(\Abar;B)_{\hat{\rho}}$. 
\end{proof}

\bigskip
In the following proposition, $K_D(\rho_{AB})$ refers to the distillable key of the state $\rho_{AB}$. For the exact definition, please refer to Definition~8 of \cite{Horodecki2009}. 
\begin{proposition}
Let $\rho_{AB}$ be a bipartite state, $\hat{\rho}_B^{a,x}$ an assemblage resulting from the action of a POVM on Alice's system, and $p(a,b|x,y)$ a quantum correlation resulting from the action of an additional POVM on Bob's system. Then, the device-independent secret-key-agreement capacity of the quantum correlation~$p$ does not exceed the one-sided device-independent secret-key-agreement capacity of $\hat{\rho}$, which in turn does not exceed the distillable key of the state $\rho_{AB}$:
\begin{equation}
    DI(p) \leq SDI(\hat{\rho}) \leq K(\rho_{AB}).
\end{equation}
\end{proposition}

\begin{proof}
The proof is a consequence of the following observation: the DI secret-key-agreement protocol is a special case of the SDI secret-key-agreement protocol with the measurements on Bob's side corresponding to i.i.d.~measurements. Similarly, the SDI secret-key-agreement protocol is a special case of a secret-key-agreement protocol acting on the state $\rho_{AB}$ with the local operations on Alice's side consisting of i.i.d.~measurements. 
\end{proof}
\section{Examples}\label{section:examples}

\subsection{Device-independent protocol}\label{sec:device_independ}

We now consider a device that is characterized by the correlation $p$ which has the following quantum strategy: Alice and Bob share a two-qubit isotropic state $\omega_{AB}^p= (1-p) \Phi_{AB}+p\pi_A\otimes \pi_B$, where $\Phi_{AB} = \frac{1}{2}\sum_{i,j=0}^1\op{ii}{jj}$, and $\pi$ denotes the maximally mixed state. This state arises from sending one share of $\Phi_{AB}$ through a depolarizing channel. Alice's measurement choices $x_0$, $x_1$, and $x_2$ correspond to $\sigma_z$, $\frac{\sigma_z+\sigma_x}{\sqrt{2}}$, and $\frac{\sigma_z-\sigma_x}{\sqrt{2}}$, respectively. Bob's measurement choices $y_1$ and $y_2$ correspond to $\sigma_z$ and $\sigma_x$, respectively. The correlation resulting from this setup is then $p(a,b|x,y)$, with $x$ taking values from $\{x_0, x_1, x_2\}$, the variable $y$ taking values from $\{y_1,y_2\}$, and $a,b \in \{0,1\}$ being the measurement results.  A specific device-independent protocol was studied in \cite{Acin2007}, which was then used to obtain a lower bound on the key rate from the above specified correlation.

The secret-key rate in a  device-independent protocol is bounded from above as follows (Theorem~\ref{theorem:quantum_device_independent}):
\begin{align}
R\leq \sup_{p(x,y)}\inf_{\rho_{\Abar\Bbar XYE}}\sum_{x,y}p_{XY}(x,y)I(\Abar;\Bbar|E)_{x,y}.
\end{align}
The idea is now to consider some quantum extension of the probability distribution obtained from the black box, and then bound the quantum intrinsic non-locality from above. 

The technique presented below is similar to the technique used in \cite{Goodenough2016} to obtain upper bounds on the squashed entanglement of a depolarizing channel. An isotropic state is Bell local if $p \geq 1-\frac{1}{\sqrt{2}}$ \cite{Horodecki95}. This implies that the quantum intrinsic non-locality of a correlation derived from $\omega^p_{AB}$ is equal to zero for $p \geq 1-\frac{1}{\sqrt{2}}$ (Proposition~\ref{prop:prof_faithful}). For $\epsilon\leq p \leq 1-\frac{1}{\sqrt{2}}$, we can write the probability distribution $q_{\omega^p}(a,b|x,y)$ obtained from $\omega^p_{AB}$ as a convex combination of probability distributions obtained from $\omega^{\epsilon}$ and $\omega^{1-1/\sqrt{2}}$. That is, for some $0\leq \alpha \leq 1$, we have
\begin{align}\label{eq:expansion_non_locality}
q_{\omega^p}(a,b|x,y)&= (1-\alpha(\epsilon))q_{\omega^{\epsilon}}(a,b|x,y)+\alpha(\epsilon) q_{\omega^{1-1/\sqrt{2}}}(a,b|x,y).
\end{align}
By simple algebra, we obtain
\begin{equation}
\alpha(\epsilon) = \frac{p-\epsilon}{1-\frac{1}{\sqrt{2}}-\epsilon}.
\end{equation}
Equation~\eqref{eq:expansion_non_locality} can be written as
\begin{align}
q_{\omega^p}(a,b|x,y)= (1-\alpha(\epsilon))q_{\omega^{\epsilon}}(a,b|x,y)+\alpha(\epsilon) \sum_{\lambda}p(\lambda)q_{\omega^{1-1/\sqrt{2}}}(a,|x,\lambda)q_{\omega^{1-1/\sqrt{2}}}(b,|y,\lambda).
\end{align}
Then, from convexity of quantum intrinsic non-locality (Proposition~\ref{prop:quantum_convexity_intrinsic}), we obtain
\begin{align}
N^Q(\Abar;\Bbar)_{q_{\omega^p}}\leq(1-\alpha(\epsilon))N^Q(\Abar;\Bbar)_{q_{\omega^{\epsilon}}} .
\end{align}
Since the above equation is true for all $\alpha$, we find that
\begin{align}
N^Q(\Abar;\Bbar)_{q_{\omega^p}}\leq\min_{0\leq\epsilon\leq p}(1-\alpha(\epsilon))N^Q(\Abar;\Bbar)_{q_{\omega^{\epsilon}}}. 
\end{align}
This implies that
\begin{align}
N^Q(\Abar;\Bbar)_{q_{\omega^p}}\leq \min_{0\leq\epsilon\leq p}(1-\alpha(\epsilon))\sup_{p(x,y)}\inf_{\rho_{\Abar\Bbar XYE}(\epsilon)}\sum_{x,y}p(x,y)I(\Abar;\Bbar|E)_{\rho_{\Abar\Bbar E}^{x,y}(\epsilon)},
\end{align}
where $q_{\omega}^{\epsilon}$ in encoded in $\rho_{\Abar \Bbar XY}(\epsilon)$ with $\rho_{\Abar \Bbar XYE}(\epsilon)$ as the quantum extension. Let us choose a trivial extension of the state $\rho_{\Abar\Bbar }^{x,y}(\epsilon)$. It is easy to see that 
\begin{equation}
I(\Abar;\Bbar)_{\rho_{\Abar\Bbar}^{0,1}(\epsilon)}\geq I(\Abar;\Bbar)_{\rho_{\Abar\Bbar}^{x,y}(\epsilon)}\quad \forall x \in \mathcal{X},y \in \mathcal{Y}.
\end{equation}
Therefore, 
\begin{align}\label{eqn:upper_bound_device}
R\leq \min_{0\leq\epsilon\leq p}(1-\alpha(\epsilon)) I(\Abar;\Bbar)_{\rho_{\Abar\Bbar}^{0,1}(\epsilon)}
&=\min_{0\leq\epsilon\leq p}(1-\alpha(\epsilon))  \left(\frac{2-\epsilon}{2}\log_2 (2-\epsilon)+ \frac{\epsilon}{2}\log_2 \epsilon\right). 
\end{align}
We plot this upper bound in Figure~\ref{fig:intrinsic_non_local}, and we interpret it and explain the relative entropy of entanglement bound in the next subsection. 
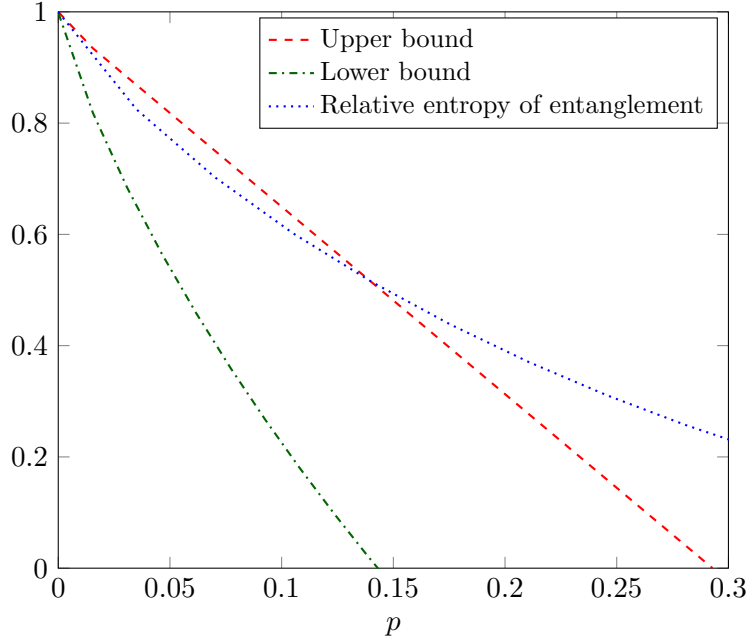
\begin{figure}
\centering
\begin{tikzpicture} 
	\begin{axis}[
	scale = 1.3,
	xlabel=$p$,
	xmin = 0,
	xmax = 0.3,
	ymin = 0,
	ymax = 1,
	tick label style={/pgf/number format/fixed, /pgf/number format/precision=3},
	legend style = {at = {(.3,0.99)},anchor = north west}, 
	legend cell align = left,]
	\addplot[thick,dashed,red] table[x=u3,y=intrinsic_nonlocal,col sep=comma] {data.txt};
    \addplot[thick,color=dgreen, dashdotted] table[x=u3,y=ratelower,col sep=comma] {data.txt};
	\addplot[thick,blue,dotted] table[x=u2,y=rel_entropy_bound,col sep=comma] {data.txt};
	\legend{\small{Upper bound},\small{Lower bound},\small{Relative entropy of entanglement}};
	\end{axis}
	\end{tikzpicture}
   \caption{In this figure, we plot the upper bound in \eqref{eqn:upper_bound_device} and the lower bound from \cite{Acin2007} for the device-independent protocol described in Section~\ref{sec:device_independ}. The relative entropy of entanglement of a qubit-qubit isotropic state is given in \cite{Vedral97}. For further explanation of this plot, see the next section. 
   \label{fig:intrinsic_non_local}}
    \end{figure}

\subsection{One-sided device-independent protocol}

\label{sec:one-sided-device-inde}

Let us now consider an assemblage $\hat{\rho}(p)$ that is generated from an isotropic state, with $x_0=\sigma_z$ and $x_1=\sigma_x$, then
\begin{equation}
\begin{split}
\rho_{X\bar{A}B}(p) &=\frac{1}{4}\left(\op{0}_X \otimes\left[\op{0}_{\bar{A}}\otimes \left(\left(1-p\right)\op{0}_B +p\pi_B\right)\right]\right)\\ 
&\qquad +\frac{1}{4}\left(\op{0}_X \otimes\left[\op{1}_{\bar{A}}\otimes \left(\left(1-p\right)\op{1}_B +p\pi_B\right)\right]\right)\\
&\qquad+\frac{1}{4}\left(\op{1}_X \otimes\left[\op{0}_{\bar{A}}\otimes \left(\left(1-p\right)\op{+}_B +p\pi_B\right)\right]\right) \\ 
&\qquad+\frac{1}{4}\left(\op{1}_X \otimes\left[\op{1}_{\bar{A}}\otimes \left(\left(1-p\right)\op{-}_B +p\pi_B\right)\right]\right). 
\end{split}
\end{equation}
If $p\geq 1/2$, it is known that $\rho_{X\Abar B}$ is unsteerable \cite{Wiseman}, and therefore intrinsic steerability is zero for $p\geq \frac{1}{2}$ (\cite[Proposition~7]{Kaur2016}).  For $\epsilon\leq p \leq \frac{1}{2}$, we can write the $\rho_{X\Abar B}(p)$ as a convex combination of states $\rho_{X\Abar B}(\epsilon)$ and $\rho_{X\Abar B}(\frac{1}{2})$. That is, for some $0\leq\alpha\leq 1$
\begin{align}
\rho_{X\Abar B}(p) = (1-\alpha)\rho_{X\Abar B}(\epsilon)+\alpha\rho_{X\Abar B}\left(\tfrac{1}{2}\right). 
\end{align}
Then, by simple algebra we obtain
\begin{equation}
\alpha(\epsilon) = \frac{p-\epsilon}{\frac{1}{2}-\epsilon}.
\end{equation}
From convexity of intrinsic steerability (Proposition~10 \cite{Kaur2016}), we obtain
\begin{equation}
S(\Abar;B)_{\hat{\rho}(p)}\leq S(\Abar;B)_{\hat{\rho}(\epsilon)} .
\end{equation}
Following the same argument as before, we obtain
\begin{equation}
S(\Abar;B)_{\hat{\rho}(p)}\leq \min_{0\leq\epsilon
\leq p}(1-\alpha(\epsilon))\sup_{p_X(x)}\inf_{\rho_{\Abar BXE(\epsilon)}}\sum_{p_X(x)}p_X(x)I(\Abar;B|E)_{\rho_{\Abar BE}(\epsilon)}.
\end{equation}
Let us now choose a trivial extension of the assemblage. It is easy to see that 
\begin{align}
I(\Abar;B)_{\rho^0(\epsilon)}&= I(\Abar;B)_{\rho^1(\epsilon)}\\&= 1 + \left(\tfrac{\epsilon}{2}\right) \log\(\tfrac{\epsilon}{2}\)+\left(1-\tfrac{\epsilon}{2}\right)\log \(1-\tfrac{\epsilon}{2}\).
\end{align}
We therefore obtain
\begin{align}\label{eqn:upper_bound_one-sided_device}
S(\Abar;B)_{\rho}= \min_{0\leq\epsilon\leq p}(1-\alpha(\epsilon))\( 1 + \left(\tfrac{\epsilon}{2}\right) \log\(\tfrac{\epsilon}{2}\)+\left(1-\tfrac{\epsilon}{2}\right)\log \(1-\tfrac{\epsilon}{2}\)\).
\end{align}
We plot this bound in Figure~\ref{fig:one-sided-device-independent}.
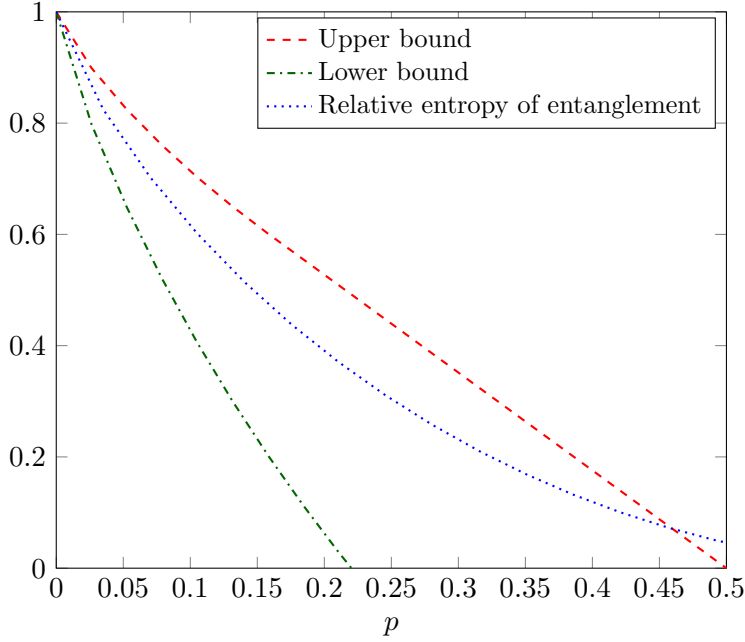
\begin{figure}
\centering
\begin{tikzpicture}
	\begin{axis}[
	scale = 1.3,
	xlabel=$p$,
	xmin = 0,
	xmax = 0.5,
	ymin = 0,
	ymax = 1,
	tick label style={/pgf/number format/fixed, /pgf/number format/precision=3},
	legend style = {at = {(.3,0.99)},anchor = north west}, 
	legend cell align = left,]
	\addplot[thick,dashed,red] table[x=u1,y=intrinsic_steer,col sep=comma] {data.txt};
    \addplot[thick,color=dgreen, dashdotted] table[x=u1,y=lower_bound_steer,col sep=comma] {data.txt};
    \addplot[thick,color=blue, dotted] table[x=u2,y=rel_entropy_bound,col sep=comma] {data.txt};
	\legend{\small{Upper bound},\small{Lower bound},\small{Relative entropy of entanglement}};
	\end{axis}
	\end{tikzpicture}
    \caption{In this figure, we plot the upper bound in \eqref{eqn:upper_bound_one-sided_device} and the lower bound from \cite{Branciard2012} for the one-sided-device-independent protocol described in Section~\ref{sec:one-sided-device-inde}. The relative entropy of entanglement of a qubit-qubit isotropic state is given in \cite{Vedral97}.
\label{fig:one-sided-device-independent}}
    \end{figure}

Due to the fact that squashed entanglement is an upper bound on the rate at which secret key can be distilled from an isotropic state \cite{CEH07,Wilde2016}, as well as the above protocols being particular protocols for secret key distillation, squashed entanglement  is also an upper bound on the rate at which the secret key can be distilled in one-sided-device-independent and device-independent protocols. However, the upper bound on squashed entanglement of an isotropic state that we obtain after choosing the extension as given in \cite{Goodenough2016} is greater than the bound obtained on intrinsic steerability of the assemblage considered above. Therefore, we do not plot the squashed-entanglement bounds in Figures~\ref{fig:intrinsic_non_local} or \ref{fig:one-sided-device-independent}. 

For the same reason given above, the relative entropy of entanglement is also an upper bound on the rate at which secret key can be distilled in one-sided-device-independent and device-independent protocols \cite{Horodecki2009}. The relative entropy of entanglement of qubit-qubit isotropic states has been calculated in \cite{Vedral97}, which we plot in the above figures. This bound performs better than intrinsic non-locality and intrinsic steerability in certain regimes.  This suggests that it might be worthwhile to explore if relative entropy of steering \cite{GL2015,KW2017} and relative entropy of non-locality \cite{DGG2005} would be useful as upper bounds for one-sided-device-independent and device-independent quantum key distribution, respectively. 

The bounds that we obtain do not closely match the lower bounds obtained from prior literature. One reason for this discrepancy can be traced back to the following question: is a violation of a Bell inequality or a steering inequality sufficient for security in DI-QKD and SDI-QKD, respectively? Since our measure is faithful, it is equal to zero if and only if there is no violation of steering inequality or Bell inequality. However, the lower bounds hit zero at a lower value of $p$ than expected from the faithfulness condition. Another possible reason for the discrepancy has been discussed in Section~\ref{sec:device_independ}, pertaining to two-way error correction that is allowed in the protocols considered above. 
 
\section{Conclusion and outlook}\label{section:conclusion}

In the present work, we have introduced information-theoretic measures of non-locality called \emph{intrinsic non-locality} and \emph{quantum intrinsic non-locality}. They are inspired by the intrinsic information \cite{Maurer1999} and have a form similar to squashed entanglement \cite{Christandl2004} and intrinsic steerability \cite{Kaur2016}. We have proven that intrinsic non-locality and quantum intrinsic non-locality are upper bounds on secret-key rates in device-independent secret-key-agreement protocols. Similarly, we have proven that restricted intrinsic steerability is an upper bound on secret-key rates in  one-sided device-independent secret-key-agreement protocols. To our knowledge, this is the first time that monotones of Bell non-locality and steering have been used to obtain upper bounds on device-independent and one-sided-device-independent secret-key rates, respectively. The faithfulness properties for intrinsic steerability and intrinsic non-locality that we have proven here are of independent interest. 

We now give an overview of the remaining open problems not addressed by the present work. It is not known if either intrinsic non-locality or intrinsic steerability are asymptotically continuous. A naive approach for establishing these properties is to follow the proof for asymptotic continuity of squashed entanglement \cite{Fannes2011}; however, this approach does not straightforwardly apply due to the no-signaling constraints on the extension system. From a foundational perspective, it would be interesting to provide an example of a probability distribution for which the intrinsic non-locality with a classical no-signaling extension is different from intrinsic non-locality with a quantum no-signaling extension. 

We also suspect that the squashed entanglement of a bipartite state $\rho_{AB}$ is greater than or equal to the restricted intrinsic steerability of an assemblage that results from measuring $\rho_{AB}$. The approach in Proposition~\ref{prop:inequality} does not apply because it does not account for the factor of $1/2$ present in the definition of squashed entanglement. 

Another promising direction to pursue is to improve the upper bounds on secret-key rates for device-independent and one-sided-device-independent protocols. Several works in the classical information theory literature have introduced modifications of classical intrinsic information \cite{Renner2003,Gohari2010} in order to obtain better bounds on secret-key rates than intrinsic information. In \cite{Renner2003}, a modified measure of intrinsic information, called reduced intrinsic information, was introduced and proved to be a better upper bound on secret-key rate than intrinsic information \cite{Maurer1999}. This bound was also subsequently improved further in \cite{Gohari2010}. It would be interesting to check if these techniques lead to improvements on the upper bounds presented by intrinsic non-locality and intrinsic steerability. 

One of the most important open questions is to determine if the relative entropy of steering \cite{GL2015,KW2017} and relative entropy of non-locality \cite{DGG2005} would be useful as upper bounds for one-sided-device-independent and device-independent secret-key-agreement protocols, respectively. It is possible that this might be the case; if true, it could lead to tighter upper bounds for certain device-independent and one-sided-device-independent protocols.

\section*{Acknowledgments}

We are grateful to Rotem Arnon-Friedman for discussions on device-independent QKD. Eneet Kaur and Mark M. Wilde acknowledge support from the US Office of Naval Research and the National Science Foundation under Grant No.~1350397. Andreas Winter acknowledges support from the ERC
Advanced Grant IRQUAT, the Spanish MINECO (project FIS2016-86681-P), with
the support of FEDER funds, and the Generalitat de Catalunya, CIRIT project 2014-SGR-966.

\bibliographystyle{alpha}
\bibliography{intrinsic_non_locality}


\end{document}